\documentclass[a4paper,10pt,reqno]{amsart}

\usepackage{hyperref}
\usepackage{amsmath,amssymb,mathrsfs}
\usepackage{enumerate}
\usepackage{tikz}
\tikzstyle{infinito}=[circle,inner sep=0pt,minimum size=0mm]
\tikzstyle{nodo}=[circle,draw,fill,inner sep=0pt,minimum size=2.5pt]
\tikzstyle{nodo2}=[circle,draw,fill,inner sep=0pt,minimum size=4pt]

\newtheorem{Theorem}{Theorem}[section]
\newtheorem{Lemma}[Theorem]{Lemma}
\newtheorem{Proposition}[Theorem]{Proposition}

\newtheorem{Remark}[Theorem]{Remark}
\newtheorem{Definition}[Theorem]{Definition}

\numberwithin{equation}{section}

\DeclareMathOperator{\diag}{diag}
\DeclareMathOperator{\Ran}{Ran}
\DeclareMathOperator{\Ker}{Ker}
\newcommand{\NA}{\mathbb{N}}
\newcommand{\CO}{\mathbb{C}}
\newcommand{\RE}{\mathbb{R}}
\newcommand{\ID}{\mathbb{I}}
\newcommand{\NULL}{\mathbb{O}}
\newcommand{\la}{\lambda}
\newcommand{\ve}{\varepsilon}
\newcommand{\CC}{\mathcal{C}}
\newcommand{\EE}{\mathcal{E}}
\newcommand{\GG}{\mathcal{G}}
\newcommand{\HH}{\mathcal{H}}
\newcommand{\OO}{\mathcal{O}}
\newcommand{\RR}{\mathcal{R}}
\newcommand{\VV}{\mathcal{V}}
\newcommand{\II}{{\bf 1}}

\renewcommand{\Im}{\operatorname{Im}}
\renewcommand{\Re}{\operatorname{Re}\,}

\newcommand{\OD}{{\bf \Theta}}
\newcommand{\KK}{K}
\newcommand{\Kcal}{\mathcal{K}}
\newcommand{\MM}{M}
\newcommand{\uch}{\underline{\hat c}}
\newcommand{\WHC}{\widehat{C}}
\newcommand{\WHH}{\widehat{H}}
\newcommand{\WHP}{\widehat{P}}
\newcommand{\WHR}{\widehat{R}}
\newcommand{\zero}{{\bf 0}}

\newcommand{\uc}{\underline{c}}
\newcommand{\BB}{\mathcal B}
\newcommand{\uv}{\underline{q}}
\newcommand{\uu}{\underline{p}}

\title[]{Scale Invariant Effective Hamiltonians for a Graph with a Small Compact Core}
\author[]{Claudio Cacciapuoti}
\address{Dipartimento di Scienza e Alta Tecnologia, Sezione di Matematica, Universit\`a dell'Insubria, Via Valleggio 11, 22100 Como, Italy, EU}
\email{claudio.cacciapuoti@uninsubria.it	}%

\thanks{The author is grateful to Gregory Berkolaiko and Andrea Posilicano for enlightening discussions. The author also thanks the anonymous referees for many useful comments that helped to improve the quality of the paper.}

\date{}

\begin{document}
%%%%%%%%%%%%%%%%%%%%%%%%%%%%%%%%%
%ABSTRACT
%%%%%%%%%%%%%%%%%%%%%%%%%%%%%%%%%
\begin{abstract}We consider a compact metric graph  of size  $\varepsilon$, and attach to it several edges (\emph{leads}) of length of order one (or of infinite length).  As $\varepsilon$ goes to zero, the graph $\mathcal{G}^\varepsilon$ obtained in this way looks like the star-graph formed  by the leads joined in a central vertex. \\ 
On $\mathcal{G}^\varepsilon$ we define an Hamiltonian $H^\varepsilon$, properly  scaled with the parameter $\varepsilon$. We prove that there exists a scale invariant effective Hamiltonian on the star-graph that approximates $H^\varepsilon$ (in a suitable norm resolvent sense) as $\varepsilon\to0$. \\
The effective Hamiltonian depends on the spectral properties of an auxiliary $\varepsilon$-independent Hamiltonian defined on  the compact graph obtained by setting $\varepsilon = 1$. If zero is not an eigenvalue of the auxiliary Hamiltonian, in the limit $\varepsilon\to0$, the leads are decoupled. 
\end{abstract}

\maketitle

\begin{footnotesize}
 \emph{Keywords:} Metric graphs;  scaling limit;  Kre\u{\i}n formula; point interactions. 
 
 \emph{MSC 2010:} 	
 81Q35; %Quantum theory - Quantum mechanics on special spaces: manifolds, fractals, graphs, etc.
 47A10; %Operator theory - Spectrum, resolvent
 %47B25; %Operator theory - Symmetric and selfadjoint operators (unbounded)
 34B45. %Ordinary differential equations - Boundary value problems on graphs and networks
 \end{footnotesize}

%%%%%%
%SECTION
%%%%%%
\section{Introduction}
%Schr\"odinger operators on metric graphs are often presented as effective models for constrained systems in which the transverse dynamics is strongly confined to an almost one dimensional graph-like manifold. To prove that such an approximation is possible may be an hard task we refer to \cite{??} for a discussion on this topic. 
%
%In this paper we take one step further. In applications it may happen that several edges of the graph have characteristic length much smaller than others. To fix ideas consider a graph obtained by attaching to a compact graph of size $\ve$ several edges of length of order one (\emph{leads}). 

One nice feature of quantum graphs (metric graphs equipped with differential operators) is that they are simple objects. In many cases, for example in the framework of the analysis of self-adjoint realizations of the Laplacian,  it is possible to write down  explicit formulae for the relevant quantities, such as the resolvent or the scattering matrix (see, e.g.,  \cite{kostrykin-schrader-jpa99} and \cite{kostrykin-schrader-cm06}). 

If the graph is too intricate though, it can be difficult, if not impossible, to perform exact computations. In such a situation, one may be interested in a simpler, effective model which captures only the most essential features of a complex quantum graph.

If several edges of the graph are much shorter then others, an effective model should rely on a simpler graph obtained by shrinking  the short edges into vertices. These new vertices should keep track of at least some of the spectral or scattering properties of the shrinking  edges, and perform as a black box approximation for a small, possibly  intricate, network. 

Our goal is to understand under what circumstances this type of effective models can be implemented. In this report we give some preliminary results showing that under certain assumptions such approximation is possible. 

To fix the ideas, consider a compact metric graph $\GG^{in,\ve}$ of size (total length) $\ve$, and attach to it several edges of length of order one (or of infinite length), the \emph{leads}. Clearly, when $\ve$ goes to zero, the graph obtained in this way (let us denote it by $\GG^\ve$) looks like the star-graph formed  by the leads joined in a central vertex. Let us denote by $\GG^{out}$ such star-graph and by  $v_0$ the central vertex. 

Given a certain Hamiltonian (self-adjoint Schr\"odinger operator) $H^\ve$ on $\GG^\ve$, we want to show that there exists an Hamiltonian $H^{out}$ on $\GG^{out}$ such  that, for small $\ve$, $H^{out}$ approximates (in a sense to be specified) $H^\ve$. Of course,  one main issue is to understand what boundary conditions in the vertex $v_0$ characterize the domain of $H^{out}$. 

It turns out that, under several technical assumptions, the boundary conditions in $v_0$ are fully determined by the spectral properties of an auxiliary, $\ve$-independent  Hamiltonian defined on the graph $\GG^{in} = \GG^{in,\ve = 1}$. 

Below we briefly discuss these technical assumptions, and refer to Section \ref{s:2} for the details.
\begin{itemize}
%\begin{enumerate}[(i)]
\item[(i)] The Hamiltonian $H^\ve$ on  $\GG^\ve$ is  a self-adjoint realization of the operator $-\Delta + B^\ve$ on $\GG^\ve$, where $B^\ve$  is a potential  term. 
\item[(ii)] To set up the graph $\GG^\ve$ we select $N$ distinct vertices in $\GG^{in,\ve}$ (we call them \emph{connecting vertices}) and attach to each of them one lead, which is either a finite or an infinite length edge. The domain of $H^\ve$ is characterized by \emph{Kirchhoff} (also called \emph{standard} or \emph{free})  boundary conditions at the connecting vertices, i.e., in each connecting  vertex  functions are continuous and the sum of the outgoing derivatives equals zero. 
\item[(iii)] (Scale Invariance) The small (or \emph{inner}) part  of the graph scales uniformly in $\ve$, i.e., $\GG^{in,\ve} = \ve \GG^{in}$. The Hamiltonian $H^\ve$ has a specific scaling property with respect to the parameter $\ve$: loosely speaking, up to a multiplicative factor,  the ``restriction'' of $H^\ve$ to $\GG^{in,\ve}$ is unitarily equivalent to an $\ve$-independent  operator  on $\GG^{in}$.  The scale invariance property can be made precise by reasoning in terms of Hamiltonians on  the inner graph $\GG^{in,\ve}$. This is done in  Section 4 below. Here we just mention that this assumption forces the scaling on  the $in$ component of the potential $B^{in,\ve}(x) = \ve^{-2} B^{in} (x/\ve)$, $x\in\GG^{in,\ve}$, and, in the vertices of $\GG^{in,\ve}$,  the Robin-type vertex conditions (if any) also scale with $\ve$ accordingly.  
\item[(iv)] The  ``restriction'' of $H^\ve$ to the leads does not depend on $\ve$. In particular, $B^{out}$,  the $out$ component of the potential,  does not depend on $\ve$. 
%\end{enumerate}
\end{itemize}

We prove that it is always possible to  identify an Hamiltonian $H^{out}$ on $\GG^{out}$ that approximates the Hamiltonian $H^\ve$. The Hamiltonian $H^{out}$ is  a self-adjoint realization of the operator $-\Delta + B^{out}$ on $\GG^{out}$, and it is characterized by scale invariant vertex  conditions in $v_0$, i.e., vertex conditions with no Robin part (see  \cite[Sec. 1.4.2]{berkolaiko-kuchment13}); in our notation, scale invariant means  $\Theta_v=0$ in Eq. \eqref{H-dom}. The precise form of the possible effective Hamiltonians  is given in  Def.s \ref{d:mathringHout} and \ref{d:effect} below. 

The convergence of $H^\ve$ to $H^{out}$ is understood in the following sense. We look at the resolvent operator $R_z^\ve := (H^\ve - z)^{-1}$, $z\in\CO\backslash\RE$, as an operator in the Hilbert space $L^2(\GG^\ve) = L^2(\GG^{out})\oplus L^2(\GG^{in,\ve})$. In the limit $\ve \to 0$, the bounded operator $R_z^\ve$ converges to an operator which is diagonal in the decomposition $L^2(\GG^{out})\oplus L^2(\GG^{in,\ve})$. The $out/out$ component of the limiting operator is the resolvent of a self-adjoint operator in $L^2(\GG^{out})$,  which we identify as the effective Hamiltonian on the star-graph. 

Additionally,  we characterize the limiting boundary conditions in the vertex $v_0$ in terms of the spectral properties of an auxiliary Hamiltonian on the (compact) graph $\GG^{in}= \GG^{in,\ve=1}$. We distinguish two mutually exclusive cases: in one case (that we call \emph{Generic}) zero is not an eigenvalue of the auxiliary Hamiltonian; in the other case (we call it \emph{Non-Generic}) zero is an eigenvalue of the auxiliary Hamiltonian. 

In the Generic Case the effective Hamiltonian, denoted by $\mathring H^{out}$,  is characterized by \emph{Dirichlet} (also called \emph{decoupling}) boundary conditions in the vertex $v_0$, i.e., functions in its domain are zero in $v_0$, see Def. \ref{d:mathringHout}. From the point of view of applications this is the less interesting case, since the  leads are decoupled  (no transmission through $v_0$ is possible). 

In the Non-Generic Case the situation is more involved. If zero is an eigenvalue of the auxiliary Hamiltonian one can identify a corresponding set of orthonormal eigenfunctions (in general eigenvalues can have multiplicity larger than one, included the zero eigenvalue).  In the domain of the effective Hamiltonian $\WHH^{out}$, the boundary conditions in $v_0$ are associated to the values of these eigenfunctions in the connecting vertices, see Def. \ref{d:effect}. In this case, the boundary conditions in the vertex $v_0$ are scale invariant but, in general, not of decoupling type. For example, if the multiplicity of the zero eigenvalue is one, and the corresponding eigenfunction assumes the same value in all the connecting vertices, the boundary conditions are of Kirchhoff type. 

The proof of the convergence is based on a Kre{\u\i}n-type formula  for the resolvent $R_z^\ve$. This formula allows us to write $R_z^\ve$ as a block matrix operator in the decomposition $L^2(\GG^\ve) = L^2(\GG^{out})\oplus L^2(\GG^{in,\ve})$ (see Eq. \eqref{Rve_main}). In the formula,  the first term, $\mathring R_z^\ve$,  is block diagonal and contains the resolvents of  $\mathring H^{out}$ and  $\mathring H^{in,\ve}$ (a scaled down  version of the Auxiliary Hamiltonian, see Section \ref{ss:auxham}); the second term is non-trivial, and couples the $out$ and $in$ components to reconstruct the resolvent of the full Hamiltonian $H^\ve$.  As $\ve$ goes to zero, the off-diagonal components in  $R_z^\ve$ converge to zero, hence,  the $out$ and $in$ components are always decoupled in the limit. A careful analysis of the non-trivial term in formula \eqref{Rve_main} shows that it converges to zero in the Generic Case. In the Non-Generic Case, instead, the  $out/out$ component of the non-trivial term converges to a finite operator, and the whole   $out/out$ component  of  $R_z^\ve$ reconstructs the resolvent of the effective Hamiltonian $\WHH_0$. 

The limiting behavior of $H^\ve$ is essentially determined by the small $\ve$ asymptotics of the spectrum of the inner Hamiltonian $\mathring H^{in,\ve}$. The scale invariance assumption implies that the eigenvalues of  $\mathring H^{in,\ve}$ are given by $\la_n^\ve = \la_n/\ve^2$, where $\la_n$ are the eigenvalues of the (scaled up) auxiliary Hamiltonian $\mathring H^{in}$. Obviously, all the non-zero eigenvalues move to infinity as $\ve \to 0$; the zero eigenvalue instead, if it exists, persists, for this reason it plays a special r\^ole in the analysis. 

Closely related to our work is the paper by G. Berkolaiko, Y.  Latushkin, and S. Sukhtaiev  \cite{berkolaiko-latushkin-sukhtaiev:rxv18}, to which we refer also for additional references. In \cite{berkolaiko-latushkin-sukhtaiev:rxv18} the authors analyze the convergence of Schr\"odinger operators on metric graphs with shrinking edges. Our setting is similar to the one in \cite{berkolaiko-latushkin-sukhtaiev:rxv18} with several differences. In \cite{berkolaiko-latushkin-sukhtaiev:rxv18} there are no restrictions on the topology of the graph, i.e., $\GG^{out}$ is not necessarily a star-graph; outer edges can form loops,  be connected among them or to arbitrarily intricate finite length graphs. In \cite{berkolaiko-latushkin-sukhtaiev:rxv18}, moreover,  the scale invariance assumption is missing. With respect to our work, however, the potential terms in \cite{berkolaiko-latushkin-sukhtaiev:rxv18} do not play an essential r\^ole in the limiting problem (because they are  uniformly bounded in the scaling parameter).

As it was done in \cite{berkolaiko-latushkin-sukhtaiev:rxv18}, to analyze the convergence of $H^\ve$ to $H^{out}$, since they are operators on different Hilbert spaces, one could use the notion of $\delta^\ve$-quasi unitary equivalence (or generalized norm resolvent convergence) introduced by  P. Exner and O. Post in the series of works \cite{exner-post:05,exner-post:aip08,exner-post:09,post:06,post:book}. In Th.s \ref{t:mainG} and \ref{t:mainNG} we state our main results in terms of the expansion of the  resolvent in the decomposition $L^2(\GG^\ve) = L^2(\GG^{out})\oplus L^2(\GG^{in,\ve})$; and comment on the $\delta^\ve$-quasi unitary equivalence of the operators $H^\ve$ and $\mathring H^{out}$ (or $\WHH^{out}$) in Rem. \ref{r:EP}.  

Our analysis, with the scaling on the potential $B^{in,\ve}(x) = \ve^{-2} B^{in} (x/\ve)$,  is also related to the problem of approximating point-interactions on the real line through scaled potentials in the presence of a zero energy resonance, see, e.g., \cite{golovaty-hryniv:09}. The same type of scaling  arises naturally  also in the study of the convergence of Schr\"odinger operators in thin waveguides to operators on graphs, see, e.g., \cite{albeverio-cacciapuoti-finco:07, cacciapuoti:17, cacciapuoti-exner:07, cacciapuoti-finco-aa10}. 

Problems on graphs with a small compact core have been studied in several papers in the case in which $\GG^\ve$ is itself a star-graph, see, e.g.,   \cite{exner-manko13,exner-manko:14,manko:12,manko:14, manko:15}. In particular, in the latter  series of works, the authors point out  the r\^ole of the zero  energy eigenvalue.

Also related to our work is the problem of the approximation of  vertex conditions through ``physical Hamiltonians''. In \cite{cheon-exner-turek:10} (see also references therein), it is shown that all the possible  self-adjoint boundary conditions at the central  vertex of a star-graph, can be obtained as the  limit of Hamiltonians with $\delta$-interactions and magnetic field  terms  on a graph with a shrinking  inner part. 

Instead of looking at the convergence of the resolvent, a different approach consists in the analysis of the  time dependent problem. This is done, e.g., in \cite{mehmeti-ammari-nicaise:17}, for a tadpole-graph as the circle shrinks to a point. 

The paper is structured as follows. In Section \ref{s:2} we introduce some notation, our assumptions  and present the main results, see Th.s \ref{t:mainG} and \ref{t:mainNG}. In Section  \ref{s:3} we discuss the  Kre\u{\i}n  formulae for the resolvents of  $H^\ve$ and $\WHH^{out}$ (the limiting Hamiltonian in the Non-Generic Case). These formulae are the main tools in our analysis. In Section \ref{s:4} we discuss the scale invariance properties of the auxiliary Hamiltonian, and other relevant operators. In Section \ref{s:5} we prove Th.s \ref{t:mainG} and \ref{t:mainNG}. In doing so we present the results with a finer estimate of the remainder, see Th.s \ref{t:mainG2} and \ref{t:mainNG2}. We conclude the paper with two  appendices: in Appendix \ref{a:krf} we briefly discuss the proofs of the Kre\u{\i}n resolvent formulae from Section \ref{s:3}; in Appendix \ref{s:appB} we prove some useful bounds on the eigenvalues and eigenfunctions of $\mathring H^{in}$. 

\subsection*{Index of notation} For the convenience of the reader we recall here the notation for  the Hamiltonians used  in our analysis. For the definitions we refer to Section \ref{s:2} below. 
\begin{itemize}
\item $H^\ve$ full Hamiltonian.
\item $\mathring H^{in}$ auxiliary Hamiltonian
\item  $\mathring H^{in,\ve}$ scaled down auxiliary Hamiltonian (see Definition \ref{d:Hinve0} and Section \ref{s:4}).  $\mathring H^{in} = \mathring H^{in,\ve =1}$.
\item $\mathring H^{out}$ effective Hamiltonian in the Generic Case.
\item $\widehat H^{out}$ effective Hamiltonian in the Non-Generic Case. 
\item $ \mathring H^\ve$ diagonal Hamiltonian $ \mathring H^\ve =\diag(\mathring H^{out}, \mathring H^{in,\ve})$ in the decomposition $L^2(\GG^\ve) = L^2(\GG^{out})\oplus L^2(\GG^{in,\ve})$ (see Section \ref{s:3}).
\end{itemize} 

%%%%%%
%SECTION
%%%%%%
\section{\label{s:2}Preliminaries and main result}
For a general introduction to metric graphs we refer to the monograph \cite{berkolaiko-kuchment13}. Here, for the convenience of the reader, we introduce some notation and recall few basic notions that will be used throughout the paper. 

\subsection{\label{ss:general}Basic  notions and notation}
To fix the ideas we start by selecting a collection of points, the \emph{vertices} of the graph, and a connection rule among them. The bonds joining the vertices  are associated to oriented segments  and are the finite-length  \emph{edges} of the graph. Other edges can be of infinite length, these edges are connected only to one vertex and are associated to half-lines.  In this way we obtained a metric graph, see, e.g.,   Fig. \ref{f:1}. 
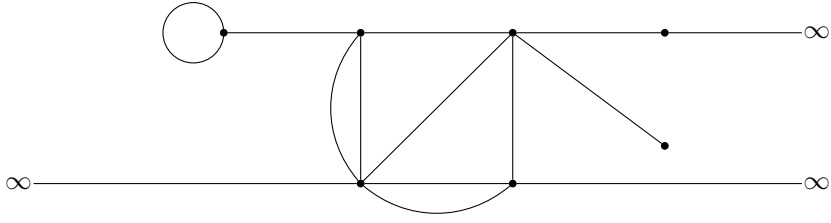
\begin{figure}[h!]
\begin{center}
\begin{tikzpicture}
\node at (-6.5,0) [infinito](0) {${\infty}$};
\node at (-2,0) [nodo] (5) {};
\node at (0,0) [nodo] (4) {};
\node at (-2,2) [nodo] (2) {};
\node at (0,2) [nodo] (3) {};
\node at (4,0) [infinito] (8) {${\infty}$};
\node at (4,2) [infinito] (7) {${\infty}$};
\node at (2,2) [nodo] (6) {};
\node at (-3.8,2) [nodo] (1) {};
%\node at (-2.2,2.2){$v_2$};
%\node at (0.1,2.2){$v_3$};
%\node at (0.2,-0.2){$v_4$};
%\node at (-2.2,-0.2){$v_5$};
%\node at (2,2.2){$v_6$};
%\node at (-3.6,2.2){$v_1$};
\draw [-,black] (-4.2,2) circle (0.4cm) ;
\draw [-,black] (5) -- (2);
\draw [-,black] (0) -- (5);
\draw [-,black] (5) -- (4);
\draw [-,black] (5) to [out=-40,in=-140] (4);
\draw [-,black] (5) to [out=130,in=-130] (2);
\draw [-,black] (5) -- (3);
\draw [-,black] (1) -- (2);
\draw [-,black] (2) -- (3);
\draw [-,black] (3) -- (6);
\draw [-,black] (4) -- (3);
\draw [-,black] (6) -- (7);
\draw [-,black] (4) -- (8);
\node at (2,0.5) [nodo] (10) {};
%\node at (2,0.3){$v_7$};
\draw [-,black] (3) -- (10);
\end{tikzpicture}
\caption{A metric graph with  7 vertices (marked by dots) and  14 edges (3 of which are half-lines).}
\label{f:1}
\end{center}
\end{figure}
Given a metric graph $\GG$ we denote by $\EE$ the set of its edges and by $\VV$ the set of its vertices. We shall also use the notation $|\EE|$ and $|\VV|$ to denote the cardinality of $\EE$ and $\VV$ respectively. We shall always assume that both $|\EE|$ and $|\VV|$ are finite. 

For any $e\in\EE$, we identify the corresponding edge with  the segment $[0,\ell_e]$ if $e$ has finite length $\ell_e>0$, or with $[0,+\infty)$  if $e$ has infinite length. 

Given  a function $\psi:\GG \to \CO$, for $e\in\EE$,  $\psi_e$ denotes its  restriction to  the  edge $e$. With this notation in mind one can define the Hilbert space 
\[\HH:= \bigoplus_{e\in\EE} L^2(e),\]
with   scalar product and  norm  given by 
\[
(\phi,\psi)_\HH := \sum_{e\in\EE} (\phi_e,\psi_e)_{L^2(e)} \qquad \text{and} \qquad 
\|\psi\|_\HH := \left(\psi,\psi\right)_\HH^{1/2}. 
\]
%%In a similar way one can define the spaces $L^p(\GG)\equiv  \bigoplus_{e\in\EE} L^p(e)$, $p\in[1,+\infty]$, equipped with the  norms 
%%\[
%%\|\psi\|_\infty := \sup_{e\in \EE} \|\psi_e\|_{L^\infty(e)}\qquad \text{and} \qquad \|\psi\|_p := \left(\sum_{e\in\EE} \|\psi_e\|_{L^p(e)}^p\right)^{1/p} \quad p\in[1,+\infty). 
%%\]
In a similar way one can define the  Sobolev space $\HH_2:= \bigoplus_{e\in\EE} H^2(e)$,  equipped with the norm 
\[
 \|\psi\|_{\HH_2 } := \left(\sum_{e\in\EE} \|\psi_e\|_{H^2(e)}^2\right)^{1/2}. 
\]
Note that functions in $\HH_2$ are continuous in the edges of the graph but  do not need to be continuous in the vertices. 

For any vertex $v\in\VV$ we denote by $d(v)$ the \emph{degree} of the vertex, this is the  number of edges having one endpoint identified by $v$, counting  twice the edges that have both endpoints coinciding with $v$ (\emph{loops}). Let $\EE_v\subseteq\EE$ be the set of edges which are incident to the vertex $v$. For any vertex $v$ we order the edges in $\EE_v$ in an arbitrary way,  counting twice the loops.  In this way, for an arbitrary function $\psi\in \HH_2$, one can define  the vector   $\Psi(v)\in \CO^{d(v)}$ associated to the evaluation of $\psi$ in $v$, i.e., the components of $\Psi(v)$ are given by $\psi_e(0)$ or  $\psi_e(\ell_e)$, $e\in\EE_v$,  depending whether $v$ is the initial or terminal vertex of the edge $e$,  or by both values if $e$ is a loop. \\
In a similar way one can define the vector $\Psi'(v)\in \CO^{d(v)}$ with components   $\psi_e'(0)$ and   $-\psi_e'(\ell_e)$, $e\in\EE_v$. Note that in the definition of $\Psi'(v)$, $\psi_e'$ denotes the derivative of $\psi_e(x)$ with respect to $x$, and the derivative in $v$ is always taken in the outgoing direction with respect to the vertex.

We are interested in defining self-adjoint operators in $\HH$ which coincide with the Laplacian, possibly plus a potential term.

We denote by $B$ the potential term in the operator, so that $B:\GG\to \RE$ is a  real-valued  function on the  graph; and denote by $B_e$ its restriction to the edge $e$. Additionally we assume that $B$ is bounded and compactly supported on $\GG$. 

For every vertex $v\in\VV$ we define a projection $P_v: \CO^{d(v)}\to \CO^{d(v)}$ and a self-adjoint operator $\Theta_v$ in $\Ran P_v$,  both $P_v$ and $\Theta_v$ can be identified with Hermitian $d(v) \times d(v)$ matrices. 

It is well known, see, e.g., \cite{berkolaiko-kuchment13} and \cite[Example 5.2]{pos-om08}, that the operator $H_{P,\Theta}$ defined by:
\begin{equation}\label{H-dom}
D(H_{P,\Theta}):= \left\{ \psi\in \HH_2 | \, P_v^\perp \Psi(v) = 0\,;\; P_v \Psi'(v) - \Theta_v P_v \Psi(v) = 0 \quad \forall v\in \VV  \right\}
\end{equation}
\begin{equation}\label{H-act}
(H_{P,\Theta} \psi)_e := -\psi_e'' + B_e \psi_e \qquad \forall e\in \EE 
\end{equation}
is self-adjoint. Instead of Eq.  \eqref{H-act},  we shall write  
\begin{equation}\label{H-act2}
H_{P,\Theta} \psi  := -\psi'' + B \psi ,
\end{equation}
to be understood componentwise. 

We remark that for every $P_v$ and $\Theta_v$ as above, $H_{P,\Theta}$ is a self-adjoint extension of the symmetric operator $H_{min}$
\[ 
D(H_{min}):=\left\{ \psi\in\HH_2 | \,  \Psi(v) = 0\,;\; \Psi'(v) = 0 \quad \forall v\in \VV  \right\} \qquad H_{min}\psi := -\psi'' + B \psi. 
\] 

%%%%%%%%%
%SUBSECTION
%%%%%%%%%
\subsection{Graphs with a small compact core}

We consider a graph $\GG^\ve$ obtained by attaching several  edges to a small compact core (a compact metric graph  of size $\ve$). 

We denote  the compact core of the graph by $\GG^{in,\ve}$. The graph $\GG^{in,\ve}$ is obtained by shrinking  a  compact graph $\GG^{in}$ by means of  a parameter $0<\ve<1$, more precisely, we set 
\begin{equation}\label{2.3a}
\GG^{in,\ve}= \ve \GG^{in}.
\end{equation}

We denote by $\EE^{in}$  the set of edges of the graph $\GG^{in}$ and by  $\EE^{in,\ve}$ the set of edges of the graph $\GG^{in,\ve}$.

In the graph $\GG^{in}$ (or, equivalently, in $\GG^{in,\ve}$) we select  $N$ distinct vertices that we label with $v_1, . . . , v_N$, and refer to them as  \emph{connecting vertices}. We shall denote by $\CC$ the set of connecting vertices. We denote by $\VV^{in}$ the set of all the remaining vertices, and call the elements of $\VV^{in}$ \emph{inner vertices} (note that the set $\VV^{in}$ may be empty).  

%For any vertex $v\in \CC$, $d^{in}(v)$ denotes its degree as a vertex of the graph $\GG^{in,\ve}$. 

%We denote by $\EE^{in}$ and $\VV^{in}$ the sets of edges and  vertices of the graph $\GG^{in}$ respectively. Since squeezing $\GG^{in}$ does not affect its vertices,  we identify the set of vertices of the graph $\GG^{in,\ve}$ with $\VV^{in}$.  We shall always assume that $|\VV^{in}|\geq N$. On the other hand we shall denote by $\EE^{in,\ve}$ the set of edges of the graph $\GG^{in,\ve}$. 
%
%In $\VV^{in}$ we identify $N$ vertices, that we label with $v_1, . . . , v_N$, and refer to them as  \emph{connecting vertices}. We shall denote by $\CC \subseteq \VV^{in}$ the set of connecting vertices. 

To construct the graph $\GG^\ve$, we attach to each connecting vertex one  additional edge which can be an half-line or an edge of finite length (not dependent on $\ve$). We shall call these additional edges \emph{outer edges} and denote by $\EE^{out}$  the corresponding set of edges; obviously $|\EE^{out}|= N$.  When needed, we shall  denote these edges by $e_1,...,e_N$,  so that the edge $e_j$ is connected to the vertex $v_j$, $j=1,...,N$. Moreover we shall use the notation 
\[\psi_{e_j}\equiv \psi_j\qquad e_j\in\EE^{out},\;j=1,...,N.\]
\\ 
Note that if $e\in\EE^{out}$ is of finite length the endpoint which does not coincide with the connecting vertex is of degree one (all the finite length outer edges are \emph{pendants}). \\
We shall always assume, without loss of generality, that for each edge in $\EE^{out}$ the connecting vertex is identified by $x = 0$. 

We denote by $\EE^{\ve}$ and $\VV$ the sets of edges and vertices of the graph $\GG^\ve$. We note that $\EE^{\ve} = \EE^{out}\cup\EE^{in,\ve}$ and $\VV = \VV^{out} \cup \CC \cup \VV^{in}$, where $\VV^{out}$ is the set of vertices in $\GG^\ve$ which are neither connecting nor inner vertices. 
\begin{Remark}\label{r:degree}
For any $v\in\CC$ we denote by $d^{in}(v)$ its degree as a vertex of the graph $\GG^{in,\ve}$, so that its degree as a vertex of the graph $\GG^\ve$ is $d(v) = d^{in}(v)+1$. 
\end{Remark}
As $\ve\to 0$, the inner graph shrinks to one point, in the limit all the connecting vertices merge in  one vertex which we identify with the point $x_j=0$, $x_j$ being the coordinate along the edge $e_j\in \EE^{out}$, $j=1,\dots,N$.  In the limit the graph $\GG^\ve$ looks like a star-graph with $N$ edges connected in the origin, see Fig. \ref{f:2}; we  denote the star-graph  by $\GG^{out}$.
\begin{figure}[h!]
\begin{center}
\begin{tikzpicture}
\node at (-6.5,0) [infinito](0) {${\infty}$};
\node at (-2,0) [nodo2] (5) {};
\node at (0,0) [nodo2] (4) {};
\node at (-2,2) [nodo] (2) {};
\node at (0,2) [nodo2] (3) {};
\node at (4,0) [infinito] (8) {${\infty}$};
\node at (4,2) [infinito] (7) {${\infty}$};
\node at (2,2) [nodo2] (6) {};
\node at (-3.8,2) [nodo] (1) {};
%\node at (-2.2,2.2){$v_2$};
%\node at (0.1,2.2){$v_3$};
%\node at (0.2,-0.2){$v_4$};
%\node at (-2.2,-0.2){$v_5$};
%\node at (2,2.2){$v_6$};
%\node at (-3.6,2.2){$v_1$};
\draw [-,black,dashed] (-4.2,2) circle (0.4cm) ;
\draw [-,black,dashed] (5) -- (2);
\draw [-,black] (0) -- (5);
\draw [-,black,dashed] (5) -- (4);
\draw [-,black,dashed] (5) to [out=-40,in=-140] (4);
\draw [-,black,dashed] (5) to [out=130,in=-130] (2);
\draw [-,black,dashed] (5) -- (3);
\draw [-,black,dashed] (1) -- (2);
\draw [-,black,dashed] (2) -- (3);
\draw [-,black,dashed] (3) -- (6);
\draw [-,black,dashed] (4) -- (3);
\draw [-,black] (6) -- (7);
\draw [-,black] (4) -- (8);
\node at (2,0.5) [nodo] (10) {};
%\node at (2,0.3){$v_7$};
\draw [-,black] (3) -- (10);
\end{tikzpicture}
\caption{The dashed lines represent the edges of $\GG^{in,\ve}$, the large dots the connecting vertices. The graph $\GG^{out}$ is obtained by merging  the connecting vertices. In the example in the picture, $\GG^{out}$ has three infinite edges and  one edge of finite length. }
\label{f:2}
\end{center}
\end{figure}
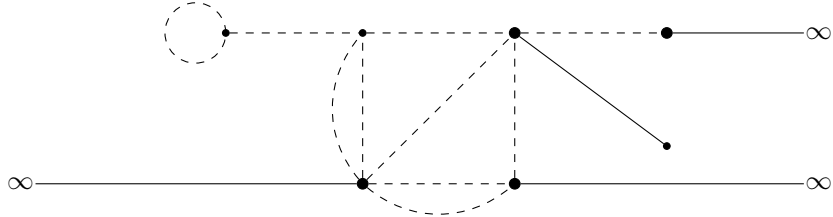

We define the Hilbert spaces:
\[\HH^\ve:= \bigoplus_{e\in\EE^\ve} L^2(e), \qquad  \HH^{out}:= \bigoplus_{e\in\EE^{out}} L^2(e) ,\qquad  \HH^{in,\ve}:= \bigoplus_{e\in\EE^{in,\ve}} L^2(e).\]
We remark that one can always think of $\HH^{\ve} $ as the direct sum 
\begin{equation}\label{outin}
\HH^{\ve} = \HH^{out}\oplus \HH^{in, \ve},
\end{equation}
and decompose each function  $\psi\in \HH^{\ve} $ as $\psi = (\psi^{out},\psi^{in})$ with $\psi^{out} \in \HH^{out}$ and $\psi^{in}\in \HH^{in,\ve}$. When no misunderstanding is possible, we omit the dependence on $\ve$, moreover we simply write $\psi$, instead of $\psi^{out}$ or $\psi^{in}$.

In a similar way we introduce the Sobolev spaces 
\[\HH_2^\ve:= \bigoplus_{e\in\EE^\ve} H^2(e), \qquad  \HH_2^{out}:= \bigoplus_{e\in\EE^{out}} H^2(e) ,\qquad  \HH_2^{in,\ve}:= \bigoplus_{e\in\EE^{in,\ve}} H^2(e).\]

\subsection{Full Hamiltonian}
Next we define an Hamiltonian $H^\ve$ in $\HH^\ve$ (of the form given in Eq.s. \eqref{H-dom} - \eqref{H-act2}); this is the object of our investigation. 

\begin{itemize}
\item Recall that if $v\in\VV^{out}$, then $d(v) =1$.  For any $v\in\VV^{out}$ we fix an orthogonal projection $P^{out}_v:\CO \to \CO$, and a self-adjoint operator $\Theta^{out}_v$ in $\Ran(P^{out}_v)$. Since vertices in $\VV^{out}$ have degree one,   $P^{out}_v$ is either $1$ or $0$; whenever  $P^{out}_v =1 $ it makes sense to define $\Theta^{out}_v$ which turns out to be the operator acting as the multiplication by a real  constant. In other words, the boundary conditions in $v\in\VV^{out}$ (of the form given in the definition of $D(H_{P,\Theta})$) can be of Dirichlet type, $\psi_e (v) = 0$; of Neumann type $\psi_e'(v)=0$; or of Robin type $\psi'_e(v) = \alpha \psi_e(v)$ with $\alpha\in\RE$.\\
It would be possible to consider a more general setting in which the outer graph has a non trivial topology, in same spirit of the work \cite{berkolaiko-latushkin-sukhtaiev:rxv18}, but we will not pursue this goal. 
\item For any $v\in \CC$ we define the orthogonal projection (see Rem. \ref{r:degree} for the definition of $d(v)$): 
\[\KK_v : \CO^{d(v)}\to \CO^{d(v)},\qquad \KK_v :=  \II_{d(v)} \left(   \II_{d(v)}, \ \cdot \ \right)_{\CO^{d(v)}} \qquad \forall v\in\CC, \] 
where $\II_{d(v)}$ denotes the vector (of unit norm) in $\CO^{d(v)}$ defined by $\II_{d(v)} = (d(v))^{-1/2}(1,...,1)$. In a similar way, we define the orthogonal projection 
\[\KK_v^{in} : \CO^{d^{in}(v)}\to \CO^{d^{in}(v)},\qquad \KK_v^{in} :=  \II_{d^{in}(v)} \left(   \II_{d^{in}(v)}, \ \cdot \ \right)_{\CO^{d^{in}(v)}} \qquad \forall v\in\CC, \]
where $\II_{d^{in}(v)}\in\CO^{d^{in}(v)}$ is defined by $\II_{d^{in}(v)} = (d^{in}(v))^{-1/2}(1,...,1)$. Both $\KK_v$ and $\KK_v^{in}$ have one-dimensional range given by the span of the vectors $\II_{d(v)}$ and $\II_{d^{in}(v)}$ respectively. \\ 
A function $\psi$ satisfies Kirchhoff conditions in the vertex $v$ (it is continuous in $v$ and the sum of the outgoing derivatives in $v$ equals zero) if and only if  $\KK_v^\perp \Psi(v) = 0$ and $  \KK_v \Psi'(v)=0$.
\item For any $v\in\VV^{in}$ we fix an orthogonal projection $P^{in}_v:\CO^{d(v)} \to \CO^{d(v)}$, and a self-adjoint operator $\Theta^{in,\ve}_v$ in $\Ran(P^{in}_v)$. 
\item We fix an $\ve$-dependent real-valued function $B^\ve : \GG^\ve \to \RE$, such that in the  $out/in$ decomposition \eqref{outin} one has  $B^\ve = (B^{out},B^{in,\ve})$. With $B^{out}:\GG^{out} \to \RE$  bounded and compactly supported. 
\item (Scale Invariance) Recall that $\GG^{in,\ve}= \ve \GG^{in}$, see Eq. \eqref{2.3a}.  We assume additionally:  that $B^{in,\ve}(x) = \ve^{-2}B^{in}(x/\ve)$, where $B^{in}:\GG^{in} \to \RE$ is bounded; and  that $\Theta^{in,\ve}_v = \ve^{-1}\Theta^{in}_v$, for all $v\in\VV^{in}$. For a discussion on the meaning and the main consequences of these assumptions we refer to Section \ref{s:4}. 
\end{itemize}

\begin{Definition}[Hamiltonian $H^\ve$]
We denote by $H^\ve$ the self-adjoint operator in $\HH^\ve$  defined by 
\begin{equation*}\label{Hve-dom}
\begin{aligned}
D(H^\ve):= \big\{ \psi\in \HH_2^\ve  | \,& {P^{in}_v}^\perp \Psi(v) = 0\,,\; P^{in}_v \Psi'(v) - \Theta^{in,\ve}_vP^{in}_v \Psi(v) = 0 \quad \forall v\in \VV^{in} ;\\
								 &{P^{out}_v}^\perp \Psi(v) = 0\,,\; P^{out}_v \Psi'(v) - \Theta^{out}_vP^{out}_v \Psi(v) = 0 \quad \forall v\in \VV^{out}; \\
								  &\KK_v^\perp \Psi(v) = 0\,,\; \KK_v \Psi'(v)  = 0 \quad \forall v\in \CC\big\}
\end{aligned}
\end{equation*}
\begin{equation*}\label{Hve-act}
H^\ve \psi := -\psi'' + B^\ve  \psi \qquad \forall \psi \in  D(H^\ve). 
\end{equation*}
\end{Definition}

\begin{Remark}In the $out/in$ decomposition  one has  
\[
(H^\ve\psi)^{out} = -{\psi^{out}}'' + B^{out} \psi^{out} 
\]
\[
(H^\ve\psi)^{in} = -{\psi^{in}}'' + B^{in,\ve} \psi^{in}.
\]
Note that the action of the outer component of $H^\ve$ does not depend on $\ve$. 
\end{Remark}

\begin{Remark}\label{r:Kv}
By the definition of $K_v$, in each connecting vertex  boundary conditions in $D(H^\ve)$ are of Kirchhoff-type:  the function $\psi$ is continuous in $v\in\CC$ and 
\[\sum_{e\sim v} \psi'_e(v) = 0 \qquad  v\in \CC, \]
where the sum is taken on all the edges incident on $v$ (counting loops twice) and the derivative is understood in the outgoing direction from the vertex. 
\end{Remark}

\subsection{Auxiliary Hamiltonian\label{ss:auxham}}
We are interested in the limit of the operator $H^\ve$ as $\ve\to0$. We shall see that the limiting properties of $H^\ve$ are strongly  related to spectral properties of the Hamiltonian $\mathring H^{in,\ve}$:
\begin{Definition}[Auxiliary Hamiltonian, scaled down version]\label{d:Hinve0}%\label{d:Hinve}
  \begin{equation}\label{DHinve}
\begin{aligned}
D(\mathring H^{in,\ve}):= \big\{ \psi\in\HH_2^{in,\ve}| \,&  {P^{in}_v}^\perp \Psi(v) = 0\,,\; P^{in}_v \Psi'(v) - \Theta^{in,\ve}_vP^{in}_v \Psi(v) = 0 \quad \forall v\in \VV^{in} ;\\
                                                                         &{ \KK^{in}_v}^\perp \Psi(v) = 0\,,\; \KK^{in}_v \Psi'(v)  = 0 \quad \forall v\in \CC\big\}
\end{aligned}
\end{equation}
\begin{equation*}\label{Hinve}
\mathring{H}^{in,\ve} \psi := -\psi'' + B^{in,\ve}  \psi \qquad \forall \psi \in D(\mathring H^{in,\ve}). 
\end{equation*}
\end{Definition}
Let $\HH^{in} = \HH^{in,\ve = 1}$, and define the unitary scaling group 
\[
U^{in,\ve} : \HH^{in} \to \HH^{in,\ve}\;,\qquad (U^{in,\ve} \psi^{in})(x) := \ve^{-1/2} \psi^{in}(x/\ve);
\]
its inverse is 
\[
{U^{in,\ve}}^{-1} : \HH^{in,\ve} \to \HH^{in}\;,\qquad ({U^{in,\ve}}^{-1} \psi^{in})(x) = \ve^{1/2} \psi^{in}(\ve x).
\]

By the scaling properties $\Theta^{in,\ve}_v = \ve^{-1}\Theta^{in}_v$ and $B^{in,\ve}(x/\ve) = \ve^{-2}B^{in}(x)$, one infers the unitary relation
\begin{equation}\label{mathringHinve}
\mathring H^{in,\ve} = \ve^{-2}  {U^{in,\ve}}\mathring H^{in}{U^{in,\ve}}^{-1} 
\end{equation}
with $\mathring H^{in}$ defined on $\HH^\ve$ and given by $\mathring H^{in} = \mathring H^{in,\ve=1}$.     One consequence of Eq. \eqref{mathringHinve} is that the spectrum of $\mathring H^{in,\ve}$ is related to the spectrum of $\mathring H^{in}$ by the relation $\sigma(\mathring H^{in,\ve}) =\ve^{-2} \sigma(\mathring H^{in})$ (see Section \ref{s:4} for more comments on the implications of the scale invariance assumption). For this reason,  we prefer to formulate the results in terms of the spectral properties of the $\ve$-independent Hamiltonian $\mathring H^{in}$ instead of  the spectral properties of $\mathring H^{in,\ve}$. 

\begin{Definition}[Auxiliary Hamiltonian $\mathring{H}^{in}$]\label{d:Haux} We call Auxiliary Hamiltonian the Hamiltonian $\mathring H^{in} = \mathring H^{in,\ve=1}$ defined on $\HH^{in}$.
% , and $\HH_2^{in} = \HH_2^{in,\ve = 1}$ be the $L^2$ and $H^2$-spaces  associated to $\GG^{in}\equiv\GG^{in,\ve = 1}$.\\ 

Letting  $\HH_2^{in} = \HH_2^{in,\ve = 1}$, the domain and action  of  $\mathring{H}^{in}$ are given by 
 \begin{equation}\label{2.4a}
\begin{aligned}
D(\mathring H^{in})= \big\{ \psi\in\HH_2^{in}| \,&  {P^{in}_v}^\perp \Psi(v) = 0\,,\; P^{in}_v \Psi'(v) - \Theta^{in}_vP^{in}_v \Psi(v) = 0 \quad \forall v\in \VV^{in} ;\\
                                                                         &{ \KK^{in}_v}^\perp \Psi(v) = 0\,,\; \KK^{in}_v \Psi'(v)  = 0 \quad \forall v\in \CC\big\}
\end{aligned}
\end{equation}
\begin{equation*}\label{Haux}
\mathring{H}^{in} \psi = -\psi'' + B^{in}  \psi \qquad \forall \psi \in D(\mathring H^{in}). 
\end{equation*}
\end{Definition}

The spectrum of $\mathring H^{in}$ consists of isolated eigenvalues of finite multiplicity, see, e.g., \cite[Th. 3.1.1]{berkolaiko-kuchment13}. For $n\in\NA$, we denote by $\lambda_n$ the eigenvalues of $\mathring H^{in}$ (counting multiplicity) and by $\{\varphi_n\}_{n\in\NA}$ a corresponding set of orthonormal eigenfunctions. \\

\begin{Definition}[Generic/Non-Generic Case]\label{d:GNG} In the analysis of the limit of $H^\ve$ we distinguish two cases: 
\begin{itemize}
\item[(1)] \emph{Generic (or Non-Resonant, or Decoupling) Case.}  $\lambda = 0$ is not an eigenvalue of the operator $\mathring H^{in}$.
\item[(2)] \emph{Non-Generic (or Resonant) Case.} $\lambda = 0$ is an eigenvalue of the operator $\mathring H^{in}$. \\In the Non-Generic Case we denote by $\{\hat \varphi_k\}_{k=1,\dots,m}$ a  set of (orthonormal) eigenfunctions corresponding to the zero eigenvalue. By Eq. \eqref{2.4a}, functions in  $D(\mathring H^{in})$ are continuous in the connecting vertices (see also Rem. \ref{r:Kv}). We denote by $\hat\varphi_k(v)$, $v\in\CC$, the value of $\hat\varphi_k$ in $v$, and define the vectors 
\begin{equation}\label{uchk}
\uch_k := (\hat\varphi_k(v_1), \dots, \hat\varphi_k(v_N))\in\CO^N,\qquad  k=1,\dots,m,\; v_j\in\CC, \; j=1,\dots,N.
\end{equation}
\end{itemize}
\end{Definition}

\begin{Definition}[$\WHC$ -- $\WHP$]\label{d:WHP} In the Non-Generic Case, let $\WHC$ be the operator 
\[
\WHC : = \sum_{k=1}^m \uch_k (\uch_k, \cdot)_{\CO^N}:\CO^N \to \CO^N.
\]
$\WHC$ is a bounded self-adjoint operator (it is an $N\times N$ Hermitian matrix). Denote by $\Ran \WHC\subseteq \CO^N$ and $\Ker \WHC\subseteq \CO^N$, the range and the kernel of $\WHC$ respectively. One has that the subspaces $\Ran \WHC$ and $\Ker \WHC $ are $\WHC$-invariant. Moreover, $\CO^N  = \Ran \WHC \oplus \Ker \WHC$. In what follows we denote by  $\WHP$ the orthogonal projection (Riesz projection, see, e.g.,  \cite[Section I.2]{gohberg-goldberg-kaashoek-90}) on $ \Ran(\WHC)$,  and by   $\WHP^\perp = \ID_N - \WHP$  the orthogonal projection on $\Ker(\WHC)$. 
\end{Definition}

\begin{Remark}\label{r:uch}
We note that  $\uv \in \Ker \WHC$ if and only if $(\uch_k,\uv)_{\CO^N} = 0$ for all $k=1,\dots,m$. To see that this indeed the case, observe that if $\uv \in \Ker \WHC$ then it must be $(\uv, \WHC \uv)_{\CO^N} = 0$, hence, $\sum_{k=1}^m |(\uch_k,\uv)_{\CO^N}|^2 = 0$, which in turn implies $(\uch_k,\uv)_{\CO^N} = 0$ for all $k=1,\dots,m$. The other implication is trivial.

Since $\WHP^\perp \uch_k \in \Ker \WHC$, we infer $0 = (\uch_k,\WHP^\perp\uch_k)_{\CO^N} = (\WHP^\perp\uch_k,\WHP^\perp\uch_k)_{\CO^N}=\|\WHP^\perp\uch_k\|_{\CO^N}^2$ for all $k = 1,\dots, m$; hence, $\WHP^\perp\uch_k = 0$, or, equivalently,  $\uch_k\in\Ran(\WHC)$. 
%We conclude that $(\uch_k,\WHC_0^{-1}\uch_{k'})_{\CO^N}$ is well defined for all $k,k' = 1,\dots,m$. 
\end{Remark}

\subsection{Effective  Hamiltonians} We shall see that the definition of the  limiting operator (effective Hamiltonian in $\HH^{out}$) depends on presence of a zero eigenvalue for $\mathring H^{in}$ (the occurrence of  the Generic Case vs. the Non-Generic Case).

Recall that for $\psi \in \HH^{out}$, we used $\psi_j$ to denote the component of $\psi$ on the edge  $e_j$ attached  to the connecting vertex $v_j$. Moreover, we assumed that the vertex $v_j$ is identified by $x=0$. With this remark in mind, given a function $\psi\in \HH_2^{out}$ we define the vectors 
\begin{equation*}\label{Psizero}
\Psi(\zero) := (\psi_1(0),\dots,\psi_N(0))^T \in \CO^N\;,\qquad \Psi'(\zero) := (\psi_1'(0),\dots,\psi_N'(0))^T \in \CO^N.
\end{equation*}
These correspond to $\Psi(v_0)$ and $\Psi'(v_0)$, as defined in Section \ref{ss:general}, where $v_0$ is the central vertex of the star-graph $\GG^{out}$. 

In the limit $\ve\to0$, the connecting vertices  in $\GG^{in,\ve}$  coincide, and can be identified with the vertex $v_0\equiv\zero$.

 We distinguish two possible effective Hamiltonians in $\HH^{out}$. 
\begin{Definition}[Effective  Hamiltonian, Generic Case]\label{d:mathringHout}  We denote by $\mathring{H}^{out}$ the self-adjoint operator  in $\HH^{out}$ defined by
\begin{equation}\label{DHout}
\begin{aligned}
D(\mathring{H}^{out}):= \big\{ \psi\in\HH_2^{out} |\,&  {P^{out}_v}^\perp \Psi(v) = 0\,,\; P^{out}_v \Psi'(v) - \Theta^{out}_vP^{out}_v \Psi(v) = 0 \quad \forall v\in \VV^{out}; \\ 
&\Psi(\zero) = 0 \big\}
\end{aligned} 
\end{equation}
\begin{equation*}\label{Hout}
H^{out} \psi := -\psi'' + B^{out}  \psi \qquad \forall \psi \in D(\mathring{H}^{out}). 
\end{equation*}
\end{Definition}

\begin{Definition}[Effective  Hamiltonian, Non-Generic Case]\label{d:effect} Let $\WHP$ be the orthogonal projection given in Def.  \ref{d:WHP}.  We denote by $\WHH^{out}$ the self-adjoint operator in $\HH^{out}$ defined by
\[
\begin{aligned}
D(\WHH^{out}):= \big\{ \psi\in \HH_2^{out} |\,&  {P^{out}_v}^\perp \Psi(v) = 0\,,\; P^{out}_v \Psi'(v) - \Theta^{out}_vP^{out}_v \Psi(v) = 0 \quad \forall v\in \VV^{out}; \\ 
&\WHP^\perp \Psi(\zero) = 0\,,\; \WHP\Psi'(\zero) = 0 \big\}
\end{aligned}
\]
\[
\WHH^{out} \psi := -\psi'' + B^{out}  \psi \qquad \forall \psi \in D(\WHH^{out}). 
\]
\end{Definition}

The boundary conditions in $\zero$ in the definitions of $D(\mathring{H}^{out})$ and  $D(\WHH^{out})$ are scale invariant (see \cite[Sec. 1.4.2]{berkolaiko-kuchment13}).

\subsection{Main result}In what follows $C$ denotes a generic positive constant independent on $\ve$. \\
 Given two Hilbert spaces $X$ and $Y$, we denote by $\BB(X,Y)$ (or simply by $\BB(X)$ if $X=Y$) the space of bounded linear operators from $X$ to $Y$, and by $\|\cdot\|_{\BB(X,Y)}$ the corresponding norm. For any $a\in\RE$, we use the notation $\OO_{\BB(X,Y)}(\ve^a)$ to denote a generic operator from $X$ to $Y$ whose norm is bounded by $C\ve^{a}$ for $\ve$ small enough.  

Given a bounded operator $A$ in $\HH^\ve$ we use the notation 
\begin{equation}\label{outindec}
A = \begin{pmatrix}
A^{out,out} & A^{out,in} \\ 
A^{in,out} & A^{in,in} 
\end{pmatrix}
\end{equation}
to describe its action in the $out/in$ decomposition \eqref{outin}: here $A^{u,v}:\HH^v\to\HH^u$, $u,v=out,in$, are operators defined according to
\begin{equation}\label{outindec2}
\begin{aligned}
 (A\psi)^{out} = & A^{out,out}\psi^{out} +  A^{out,in} \psi^{in}  \\ 
 (A\psi)^{in} = & A^{in,out}\psi^{out} +  A^{in,in} \psi^{in}  .
 \end{aligned}
 \end{equation}

\begin{Theorem}\label{t:mainG}Let $z\in\CO\backslash\RE$. In the Generic Case (see Def. \ref{d:GNG})
\begin{equation} \label{2.7a}
(H^\ve - z)^{-1} = 
\begin{pmatrix}
(\mathring H^{out} - z)^{-1}  & \NULL \\
\NULL & \NULL 
\end{pmatrix} +\OO_{\BB(\HH^{\ve})}(\ve) ,
\end{equation}
where the expansion has to be understood in the  $out/in$ decomposition \eqref{outindec}. 
\end{Theorem}

\begin{Theorem}\label{t:mainNG}Let $z\in\CO\backslash\RE$. In the Non-Generic Case  (see Def. \ref{d:GNG}), let $\WHC_0$ be the restriction of $\WHC$ to $\Ran \WHP$. 
\begin{itemize}
\item[(i)]%\label{t:mainNGi} 
If $\Ker \WHC \subset \CO^N$, $\WHC_0$ is invertible as an operator in $\WHP\CO^N$, and 
\begin{equation}\label{nona0}
(H^\ve - z)^{-1} = 
\begin{pmatrix}
(\WHH^{out} - z)^{-1}  & \NULL \\
\NULL & - z^{-1}\sum_{k,k' = 1}^m \left(\delta_{k,k'} -(\uch_k, \WHC_0^{-1} \uch_{k'} )_{\CO^N} \right) \hat \varphi_k^\ve(\hat \varphi_{k'}^\ve, \cdot)_{\HH^{in,\ve}} 
\end{pmatrix} +  \OO_{\BB(\HH^\ve)}(\ve^{1/2}),
\end{equation}
where the expansion has to be understood in the  $out/in$ decomposition  \eqref{outindec}. 

\item[(ii)]%\label{t:mainNGii}
If $\Ker \WHC = \CO^N$, then $\WHP = 0$, and  expansion \eqref{nona0} holds true with $\WHH^{out} = \mathring H^{out}$,   $(\uch_k, \WHC_0^{-1} \uch_{k'} )_{\CO^N} = 0$ for all $k,k'=1,\dots,m$,  and the error term  changed in $\OO_{\BB(\HH^\ve)}(\ve)$.

\item[(iii)]%\label{t:mainNGiii}
If the vectors $\uch_k$, $k=1,\dots,m$, are linearly independent, then $\left(\delta_{k,k'} -(\uch_k, \WHC_0^{-1} \uch_{k'} )_{\CO^N} \right) = 0$ for all $k,k' = 1,\dots, m$, and 
\begin{equation}\label{nona1}
(H^\ve - z)^{-1} = 
\begin{pmatrix}
(\WHH^{out} - z)^{-1}    &  \NULL \\
\NULL  & \NULL  \end{pmatrix} + \OO_{\BB(\HH^{\ve}) }(\ve^{1/2}).  
\end{equation}
\end{itemize}
\end{Theorem}

\begin{Remark}
Finer estimates on the remainders in Eq.s \eqref{2.7a} and \eqref{nona0} are given  in Th.s \ref{t:mainG2} and \ref{t:mainNG2} below. 
\end{Remark}

\begin{Remark}\label{r:EP} We recall and adapt to our setting the notion of $\delta^\ve$-quasi unitary equivalence  of operators acting on different Hilbert spaces  introduced by  P. Exner and O. Post, see in particular \cite[Sec. 3.2]{exner-post:09} and \cite[Ch. 4]{post:book}. See also \cite[Sec. 5]{berkolaiko-latushkin-sukhtaiev:rxv18}  for a discussion on the application of this approach to the analysis of operators on graphs with shrinking  edges. 

Let $J$  be the operator 
\[
J:\HH^{out} \to \HH^\ve, \qquad J\psi^{out} = (\psi^{out}, 0) \quad  \text{for all } \psi^{out} \in \HH^{out},   
\]
where $ (\psi^{out}, 0)$ is understood in the decomposition \eqref{outin}.  Its adjoint $J^*$ maps $\HH^{\ve}$ in $\HH^{out}$, and  is given by:
\[
J^*: \HH^\ve \to \HH^{out}, \qquad J^*\psi = \psi^{out} \quad \text{for all } \psi= (\psi^{out},\psi^{in})\in \HH^\ve .  
\]
Note that  $J^* J = \ID^{out}$, where $\ID^{out}$ is the identity in $\HH^{out}$. 

The operator $H^\ve$ is $\delta^\ve$-quasi unitarily equivalent to a self-adjoint  operator $H^{out}$ in $\HH^{out}$ if 
\begin{equation}\label{dqu}
\big\|(\ID - JJ^*) (H^\ve-z)^{-1}\big\|_{\BB(\HH^\ve)} \leq C \delta^\ve \quad \text{and} \quad 
\big\| J(H^{out} - z)^{-1}-  (H^\ve-z)^{-1} J\big\|_{\BB(\HH^{out},\HH^\ve)}  \leq C \delta^\ve,
\end{equation}
for some $z\in\CO\backslash \RE$. 

Note that in the decomposition \eqref{outindec2}, one has 
\[
(\ID - JJ^*) (H^\ve-z)^{-1} \psi = \big((H^\ve-z)^{-1}\big)^{in,out} \psi^{out} + \big((H^\ve-z)^{-1}\big)^{in,in} \psi^{in}  
\]
and 
\[
(J(H^{out} - z)^{-1}-  (H^\ve-z)^{-1} J) \psi^{out} =
\Big(\big((H^{out} - z)^{-1} - \big((H^\ve-z)^{-1} \big)^{out,out}\big)  \psi^{out}  , - \big((H^\ve-z)^{-1} \big)^{in,out} \psi^{out}\Big).
\]
Hence:

By Th. \ref{t:mainG}, in the Generic Case the operator $H^\ve$  is $\ve$-quasi unitarily equivalent to the operator $\mathring H^{out}$.

By Th. \ref{t:mainNG} - (iii), in the Non-Generic Case, if the vectors $\uch_k$, $k=1,\dots,m$, are linearly independent,  the operator $H^\ve$  is $\ve^{1/2}$-quasi unitarily equivalent to the operator $\WHH^{out}$. More precisely, the second condition in Eq. \eqref{dqu} always holds true, while the first one holds true only under the additional assumption that  the vectors $\uch_k$ are linearly independent. 

We refer to \cite{post:book} for a comprehensive discussion on the comparison between operators acting on different spaces. 
\end{Remark}

\section{\label{s:3}Kre{\u\i}n resolvent formulae}
In this section we introduce the main tools in our analysis: the Kre{\u\i}n-type resolvent formulae for the resolvents of $H^\ve$ and $\WHH^{out}$. The proofs are postponed to App. \ref{a:krf}. 

Given the Hilbert spaces $X^{out}$, $Y^{out}$, $X^{in}$, and  $Y^{in}$, and  a couple of operators $A^{out}: X^{out} \to Y^{out}$ and $A^{in}: X^{in} \to Y^{in}$, we denote by $A := \diag(A^{out},A^{in})$, the operator $A: X \to Y$, with $X := X^{out} \oplus X^{in}$ and  $Y: = Y^{out} \oplus Y^{in}$, acting as $A f: = (A^{out}f^{out}, A^{in}f^{in})$, for all $f=(f^{out},f^{in})\in X$, $f^{out}\in X^{out}$  and $f^{in}\in X^{in}$. 

We set 
\begin{equation}\label{DHve}
D(\mathring H^\ve) := D(\mathring H^{out}) \oplus  D(\mathring H^{in,\ve}) \qquad \text{and} \qquad  \mathring H^\ve :=\diag(\mathring H^{out}, \mathring H^{in,\ve}),
\end{equation}
with $\mathring H^{out}$ and $\mathring H^{in,\ve}$ given as in Def.s \ref{d:mathringHout} and \ref{d:Hinve0}

Given an operator $A$, we denote by $\rho(A)$ its resolvent set; the resolvent of $A$ is defined  as $(A-z)^{-1}$ for all $z\in \rho(A)$. 

For the resolvents of the relevant operators we introduce the shorthand notation 
\begin{equation}\label{Rve}
R^\ve_{z}:= (H^\ve - z)^{-1} \qquad z \in \rho(H^\ve) ;
\end{equation}
\begin{equation}\label{ringR}
 \mathring{R}^\ve_{z}:= (\mathring{H}^\ve - z)^{-1}   \qquad z \in \rho(\mathring H^\ve) = \rho(\mathring H^{out})\cap \rho(\mathring H^{in,\ve}); 
\end{equation}
\begin{equation}\label{ringRout}
\mathring{R}^{out}_{z}:= (\mathring{H}^{out} - z)^{-1}\qquad z\in\rho(\mathring H^{out}) ; \qquad \WHR^{out}_{z}:= (\WHH^{out} - z)^{-1}\qquad z\in\rho(\WHH^{out}) ; 
\end{equation}
\begin{equation}\label{ringRin}
\mathring{R}^{in,\ve}_{z}:= (\mathring{H}^{in,\ve} - z)^{-1} \qquad z\in \rho(\mathring H^{in,\ve }).
\end{equation}
Obviously, all the operators in Eq.s. \eqref{Rve} - \eqref{ringRin} are well-defined and bounded for $z\in \CO\backslash\RE$, moreover $\mathring{R}^\ve_{z} = \diag(\mathring{R}^{out}_{z},\mathring{R}^{in,\ve}_{z})$. \\

Our aim is to write the resolvent difference $R^\ve_{z} -  \mathring{R}^\ve_{z}$ in a suitable block matrix form, associated  to the off-diagonal matrix $\OD$ in Eq. \eqref{Theta}. To do  so we follow the approach of Posilicano \cite{posilicano_jfa01,pos-om08}.  All the self-adjoint extensions of the symmetric operator obtained by restricting a given self-adjoint operator to the kernel of a given map $\tau$ are parametrized by a projection ${\bf P}$ and a self-adjoint operator $\OD$ in $\Ran {\bf P}$.  We choose the reference operator $\mathring H^\ve$  and the map $\tau$ so that the Hamiltonian of interest $H^\ve$ is the self-adjoint extension parametrized by the identity and the self-adjoint operator given by the off-diagonal matrix $\OD$.  The Kre\u{\i}n formula for the resolvent difference  $R^\ve_{z} -  \mathring{R}^\ve_{z}$, see Lemma \ref{l:Rve},   is obtained within the approach from \cite{posilicano_jfa01,pos-om08}.  
\\

We define the maps:
\begin{equation}\label{tau-out}
\tau^{out} : \HH_2^{out} \to \CO^N \qquad \tau^{out}\psi :=\Psi'(\zero);
\end{equation}
\begin{equation}\label{tau-in}
\begin{aligned}
&\tau^{in} : \HH_2^{in,\ve} \to \CO^N \\
&\tau^{in}\psi :=\left( \frac{1}{\sqrt{d^{in}(v_1)}} (\II_{d^{in}(v_1)},\Psi(v_1))_{\CO^{d^{in}(v_1)}}, ..., \frac1{\sqrt{d^{in}(v_N)}}(\II_{d^{in}(v_N)},\Psi(v_N))_{\CO^{d^{in}(v_N)}}\right)^T.
\end{aligned}
\end{equation}
Moreover we set,  
\begin{equation*}\label{tau}
\tau: \HH_2^\ve = \HH_2^{out}\oplus \HH_2^{in,\ve} \to \CO^{2N} \qquad \tau  := \diag(\tau^{out},\tau^{in}).
\end{equation*}
Note that we are using the identification $\CO^{2N} = \CO^N \oplus \CO^N$. \\ 

The following maps are well-defined and bounded 
\begin{equation*}\label{vGout}
\breve{G}^{out}_{z} : \HH^{out} \to \CO^N \qquad \breve{G}^{out}_{z} := \tau^{out} \mathring{R}^{out}_{z} \qquad z\in\rho(\mathring H^{out})
\end{equation*}
and 
\begin{equation}\label{vGin}
\breve{G}^{in,\ve}_{z} : \HH^{in,\ve} \to \CO^N \qquad \breve{G}^{in,\ve}_{z} := \tau^{in} \mathring{R}^{in,\ve}_{z} \qquad z\in\rho(\mathring H^{in,\ve}). 
\end{equation}
Moreover we set 
\begin{equation*}%\label{vG}
\breve{G}^\ve_{z}: \HH^\ve =  \HH^{out}\oplus \HH^{in,\ve}\to \CO^{2N} \qquad \breve{G}^\ve_{z}  := \diag(\breve{G}^{out}_{z},\breve{G}^{in,\ve}_{z}),
\end{equation*}
for $z\in\rho(\mathring H^{out})\cap\rho(\mathring H^{in,\ve})$.  Note that $\breve{G}^\ve_{z} = \tau \mathring{R}^{\ve}_{z}$ and that all the maps above are well-defined bounded operators for $z\in\CO\backslash \RE$. 

The adjoint maps (in $\bar z$) are denoted by 
\begin{equation*}\label{Gout}
G^{out}_{z} :  \CO^N \to \HH^{out}  \qquad G^{out}_{z} := \breve{G}^{out*}_{\bar z},
\end{equation*}
\begin{equation}\label{Gin}
G^{in,\ve}_{z} : \CO^N \to  \HH^{in,\ve}  \qquad G^{in,\ve}_{z} :=  \breve{G}^{in,\ve*}_{\bar z} ,
\end{equation}
($^*$ denoting  the adjoint) and 
\begin{equation*}\label{G}
G^\ve_{z}: \CO^{2N} \to \HH^\ve \qquad G^\ve_{z}  := \breve{G}^{\ve*}_{\bar z}  . 
\end{equation*}
Obviously $ G^\ve_{z} = \diag(G^{out}_{z},G^{in,\ve}_{z})$ to be understood as an operator from  $\CO^{2N} =  \CO^N\oplus \CO^N$ to $\HH^{\ve} = \HH^{out}\oplus \HH^{in,\ve}$. \\ 

We note that, see Rem. \ref{r:A1},  $G^{out}_z: \CO^{N} \to \HH_2^{out}$ and  $G^{in,\ve}_z: \CO^{N} \to \HH_2^{in,\ve}$, for all $z\in \rho(\mathring H^{out})$ and $z\in \rho(\mathring H^{in,\ve})$ respectively, so that the maps  ($N\times N$, $z$-dependent matrices)
\begin{equation}\label{Mout}
\MM^{out}_z : \CO^N \to \CO^N,\quad \MM^{out}_z:= \tau^{out} G^{out}_z  \qquad  z\in \rho(\mathring H^{out}) 
\end{equation}
\begin{equation}\label{Min0}
\MM^{in,\ve}_z : \CO^N \to \CO^N,\quad \MM^{in,\ve}_z:= \tau^{in} G^{in,\ve}_z \qquad z\in \rho(\mathring H^{in,\ve}),
\end{equation}
are well defined. Moreover, we set 
\begin{equation}\label{M}
\MM^\ve_z : \CO^{2N} \to \CO^{2N},\quad \MM^\ve_z:=\diag(\MM^{out}_z,\MM^{in,\ve}_z)  \qquad  z\in \rho(\mathring H^{out}) \cap   \rho(\mathring H^{in,\ve}) =  \rho(\mathring H^{\ve});
\end{equation}
obviously $\MM^\ve_z =  \tau G^\ve (z)$. \\
 
In the following  Lemmata we give two Kre{\u\i}n-type resolvent  formulae: one allows to express the resolvent of $\WHH^{out}$ in terms of the resolvent of $\mathring H^{out}$; the other gives  the resolvent of $H^\ve$ in terms of the resolvent of $\mathring H^\ve$. For the proofs we refer to App. \ref{a:krf}, Section \ref{s:A1}. 
 
\begin{Lemma}\label{l:Rout} Let $\WHP$ be an orthogonal projection in $\CO^N$, and $\WHH^{out}$ and $\mathring H^{out}$ be the Hamiltonians defined according to  Def.s \ref{d:effect} and \ref{d:mathringHout}.  Then, for any $z\in \rho(\WHH^{out}) \cap \rho(\mathring H^{out})$,  the map $\WHP \MM^{out}_z\WHP: \WHP\CO^{N} \to \WHP\CO^{N}$ is invertible and  
\[
\WHR^{out}_z = \mathring{R}^{out}_z - G^{out}_z \WHP \big(\WHP \MM^{out}_z\WHP\big)^{-1}\WHP \breve{G}^{out}_z. 
\]
\end{Lemma}

\begin{Lemma}\label{l:Rve} Let $\OD$  be the $2N\times 2N$ block matrix
\begin{equation}\label{Theta}
\OD = \begin{pmatrix}
\NULL_N & \ID_N \\
\ID_N & \NULL_N
\end{pmatrix}.
\end{equation}
Then, for any $z\in \rho(H^\ve) \cap \rho(\mathring H^\ve)$,  the map $(\MM^\ve_z- \OD) : \CO^{2N} \to \CO^{2N}$ is invertible and  
\[
R^\ve_z = \mathring{R}^\ve_z - G^\ve_z \big(\MM^\ve_z-\OD\big)^{-1}\breve{G}^\ve_z . 
\]
\end{Lemma}

We conclude this section with an alternative formula for the resolvent $R^\ve_z$. We refer to App. \ref{a:krf}, Section \ref{s:A2},  for the proof. 
\begin{Lemma}\label{l:Rve2}
Let $z\in\CO\backslash\RE$, then the maps ($N\times N$, $z$-dependent  matrices)
\begin{equation}\label{MMMM}
\MM^{in,\ve}_z\MM^{out}_z- \ID_N : \CO^N \to \CO^N \qquad \text{and} \qquad \MM^{out}_z\MM^{in,\ve}_z - \ID_N : \CO^N \to \CO^N
\end{equation}
are invertible. Moreover, 
\begin{equation}\label{Rve_main}
R^\ve_z = \mathring R^\ve_z - G^\ve_z 
\begin{pmatrix}
 \big( \MM^{in,\ve}_z\MM^{out}_z- \ID_N\big)^{-1} \MM^{in,\ve}_z & \big(\MM^{in,\ve}_z \MM^{out}_z- \ID_N\big)^{-1}\\  \\ 
 \big(\MM^{out}_z\MM^{in,\ve}_z - \ID_N\big)^{-1}  &  \MM^{out}_z \big(\MM^{in,\ve}_z\MM^{out}_z - \ID_N \big)^{-1}  
\end{pmatrix} \breve G^\ve_z.
\end{equation}
\end{Lemma}

\section{\label{s:4}Scale invariance}
In this section we discuss the scale invariance properties of $\mathring H^{in,\ve}$ and  collect several  formulae concerning the operators $\mathring R^{in,\ve}_z$, $\breve G^{in,\ve}_z$, $G^{in,\ve}_z$, and $\MM^{in,\ve}_z$. 

%Define the unitary scaling group 
%\[
%U^{in,\ve} : \HH^{in} \to \HH^{in,\ve}\;,\qquad (U^{in,\ve} \psi^{in})(x) := \ve^{-1/2} \psi^{in}(x/\ve);
%\]
%its inverse is 
%\[
%{U^{in,\ve}}^{-1} : \HH^{in,\ve} \to \HH^{in}\;,\qquad ({U^{in,\ve}}^{-1} \psi^{in})(x) = \ve^{1/2} \psi^{in}(\ve x).
%\]
%\begin{Remark}\label{r:Hinve} Note that
%\begin{equation*}\label{mathringHinve}
%\mathring H^{in,\ve} = \ve^{-2}  {U^{in,\ve}}\mathring H^{in}{U^{in,\ve}}^{-1} 
%\end{equation*}
%where  $\mathring H^{in}$ is the auxiliary Hamiltonian  from Def. \ref{d:Haux} and $\mathring H^{in,\ve}$ was given in Def. \ref{d:Hinve0}. 
%\end{Remark}

Recall that we have denoted by $\lambda_n$ and $\{\varphi_n\}_{n\in\NA}$ the eigenvalues and a corresponding set of orthonormal eigenfunctions of $\mathring H^{in}$. 

The eigenvalues of $\mathring{H}^{in,\ve}$ (counting multiplicity)   and a corresponding set of orthonormal eigenfunctions are given by  
\begin{equation}\label{scaling}
\lambda^\ve_n = \ve^{-2} \lambda_n \;;\qquad \varphi_n^\ve(x) = 	\ve^{-1/2}  \varphi_n(x/\ve), 
\end{equation}
where $\la_n$ are the eigenvalues of $\mathring H^{in}$, and $\varphi_n$ the corresponding  (orthonormal)  eigenfunctions.

By the spectral theorem and by the scaling properties \eqref{scaling}, $\mathring R^{in,\ve}_z$ is given by 
\begin{equation}\label{Rinve}
\mathring R^{in,\ve}_z =  \sum_{n\in \NA} \frac{\varphi_n^\ve ( \varphi_n^\ve,\cdot)_{\HH^{in,\ve}}}{\lambda_n^\ve -z} = \ve^2 \sum_{n\in \NA} \frac{\varphi_n^\ve ( \varphi_n^\ve,\cdot)_{\HH^{in,\ve}}}{\lambda_n -\ve^2 z}. 
\end{equation}
Hence, its integral kernel can be written as  
\begin{equation}\label{get}
\mathring R^{in,\ve}_z(x,y) =\ve  \sum_{n\in \NA} \frac{\varphi_n(x/\ve) \varphi_n(y/\ve) }{\lambda_n -\ve^2 z} \qquad x,y \in \GG^{in,\ve}.  
\end{equation}

Since there exists a positive constant $C$ such that $\sup_{x\in\GG^{in}}|\varphi_n(x)|\leq C$   and  $\lambda_n \geq C n^2$ for $n$  large enough (see App. \ref{s:appB}),  the series in Eq. \eqref{get} is uniformly convergent for $x,y\in\GG^{in,\ve}$.  Hence,  we can write the operators $\breve G^{in,\ve}_z$ and $G^{in,\ve}_z$,  and the matrix $\MM^{in,\ve}_z$ in a similar way, see Eq.s \eqref{GGbrevein} and \eqref{Min} below. 

Note that, since functions in $D(\mathring H^{in,\ve})$ are continuous in the connecting vertices, the eigenfunctions  $\varphi_n^\ve$ can be evaluated in the connecting vertices, and,  by the definition of $\tau^{in}$ (see Eq. \eqref{tau-in}), one has  
\[ 
\tau^{in} \varphi_n^\ve = (\varphi^\ve_n(v_1),\dots,\varphi^\ve_n(v_N))^T.   
\]
So that,  for any eigenfunction $\varphi_n^\ve$ we can define the vector $\uc_n^\ve$ as
\[
\uc_n^\ve := \tau^{in} \varphi_n^\ve. 
\]
We note that  $ \uc_n^{\ve}= \ve^{-1/2} \uc_n $, with 
\[ 
\uc_n=\left(\varphi_n(v_1),\dots, \varphi_n(v_N)\right)^T,
\]
and  that the vectors $\uc_n$ are defined in the same way as the vectors $\uch_k$ in Eq. \eqref{uchk}. 

\begin{Remark}\label{r:4.2}
 In the Non-Generic Case, zero is an eigenvalue of $\mathring H^{in,\ve}$. We denote by $\{\hat \varphi_k^\ve\}_{k=1,\dots,m}$ the corresponding  set of (orthonormal) eigenfunctions given by $\hat \varphi_k^\ve(x) = \ve^{-1/2}\hat \varphi_k(x/\ve)$ where  $\hat \varphi_k$ are the eigenfunctions corresponding to the eigenvalue zero  of $\mathring H^{in}$. The vectors $\uch_k^\ve := \tau^{in}\hat \varphi_k^\ve$ are related to the vectors $\uch_k$ by the identity $\uch_k^\ve = \ve^{-1/2}\uch_k$.  
\end{Remark}

By the discussion above, and by the definitions \eqref{vGin}, \eqref{Gin}, and \eqref{Min0},  we obtain 
\begin{equation}\label{GGbrevein}
\breve G^{in,\ve}_z= \ve^{3/2} \sum_{n\in \NA} \frac{\uc_n (\varphi_n^\ve, \cdot)_{\HH^{in,\ve}}}{\lambda_n-\ve^2 z}\;;\qquad   
G^{in,\ve}_z=  \ve^{3/2} \sum_{n\in \NA} \frac{ \varphi_n^\ve (\uc_n, \cdot)_{\CO^N}}{\lambda_n - \ve^2 z},   
\end{equation}
and 
\begin{equation}\label{Min}
\MM^{in,\ve}_z= \ve   \sum_{n\in \NA} \frac{ \uc_n (\uc_n, \cdot)_{\CO^N}}{\lambda_n - \ve^2z}  .
\end{equation}

\section{\label{s:5}Proof of Theorems \ref{t:mainG} and \ref{t:mainNG}}

This section is devoted to the proofs of Th.s  \ref{t:mainG} and \ref{t:mainNG}. Actually, we shall prove a finer version of the results with more precise estimates of the remainders, see Th.s \ref{t:mainG2} and \ref{t:mainNG2} below. 
\begin{Remark}\label{r:RR} By Eq. \eqref{Rve_main}, it follows that,  in the $out/in$ decomposition \eqref{outindec}, the resolvent $R^\ve_z$ can be written as
\begin{equation}\label{RR0}
R^\ve_z = 
\begin{pmatrix}
\mathring R^{out}_z & \NULL \\
 \NULL &  \mathring R^{in,\ve}_z 
\end{pmatrix} 
-
\begin{pmatrix}
\RR_z^{out,out,\ve} & \RR_z^{out,in,\ve}\\ 
\RR_z^{in,out,\ve} & \RR_z^{in,in,\ve} 
\end{pmatrix} 
\end{equation}
with 
\begin{align}
\RR_z^{out,out,\ve}& = G^{out}_z \big( \MM^{in,\ve}_z\MM^{out}_z- \ID_N\big)^{-1} \MM^{in,\ve}_z \breve G^{out}_z; \label{RR1}\\ 
\RR_z^{in,out,\ve}  & = G^{in,\ve}_z \big(\MM^{out}_z\MM^{in,\ve}_z - \ID_N\big)^{-1} \breve G^{out}_z;\\ 
\RR_z^{out,in,\ve}  & = G^{out}_z  \big(\MM^{in,\ve}_z \MM^{out}_z- \ID_N\big)^{-1} \breve G^{in,\ve}_z;\\ 
\RR_z^{in,in,\ve}    & = G^{in,\ve}_z   \MM^{out}_z \big(\MM^{in,\ve}_z\MM^{out}_z - \ID_N \big)^{-1}  \breve G^{in,\ve}_z. \label{RR4}
\end{align}
Note that since $M_{z} = M_{\bar z}^* $  holds true both for the ``$out$'' and ``$in$'' $M$-matrices (see Eq. \eqref{MM}), one infers   $ \RR_{ z}^{in,out,\ve} = \RR_{\bar z}^{out,in,\ve*}$. 
\end{Remark}

\subsection{Generic Case. Proof of Th. \ref{t:mainG}.} In this section we study the limit of the relevant quantities in the Generic Case and prove Th. \ref{t:mainG}.

\begin{Proposition}\label{p:egg}Let $z\in\CO\backslash\RE$. In the Generic Case,  
\begin{equation}\label{buried0}
\mathring R_z^{in,\ve}= \OO_{\BB(\HH^{in,\ve})}(\ve^2);
\end{equation}
\begin{equation}\label{buried1}
\breve G^{in,\ve}_z =\OO_{\BB(\HH^{in,\ve},\CO^N)} (\ve^{3/2})\;;\qquad G^{in,\ve}_z =\OO_{\BB(\CO^N, \HH^{in,\ve})}(\ve^{3/2}).
\end{equation}
\end{Proposition}
\begin{proof}
We prove first Claim \eqref{buried0}. For any $\psi^{in}\in\HH^{in,\ve}$, since $\{\varphi_n^\ve\}_{n\in\NA}$ is an orthonormal set of eigenfunctions in $\HH^{in,\ve}$, and by Eq.  \eqref{Rinve}, we infer 
\[
\|\mathring R_z^{in,\ve} \psi^{in}\|_{\HH^{in,\ve}} = \ve^2 \left(\sum_{n\in \NA} \frac{|(\varphi_n^\ve, \psi^{in})_{\HH^{in,\ve}}|^2}{|\la_n -\ve^2 z|^2}\right)^{1/2} \leq C\ve^2 \|\psi^{in}\|_{\HH^{in,\ve}}, 
\]
where in the latter inequality we used the bound  $|\la_n -\ve^2 z|^{-2} \leq 4| \la_n |^{-2} \leq C$, which holds true in the Generic Case because $|\la_n -\ve^2 z|\geq |\la_n|/2\geq C$ for all $n\in\NA$ and $\ve$ small enough.   

To prove the first claim in Eq. \eqref{buried1}, let $\psi^{in}\in \HH^{in,\ve}$, then  
\[
\breve G^{in,\ve}_z \psi^{in}= \ve^{3/2} \sum_{n\in \NA} \frac{\uc_n (\varphi_n^\ve, \psi^{in})_{\HH^{in,\ve}}}{\lambda_n-\ve^2 z}.
\]
Hence, from the Cauchy-Schwarz inequality, 
\[\begin{aligned}
\|\breve G^{in,\ve}_z\psi^{in}\|_{\CO^N} \leq &\ve^{3/2} \sum_{n\in \NA} \frac{\|\uc_n\|_{\CO^N}\, |(\varphi_n^\ve, \psi^{in})_{\HH^{in,\ve}}|}{|\lambda_n-\ve^2 z|} \\ 
 \leq & \ve^{3/2} \|\psi^{in}\|_{\HH^{in,\ve}}  \left( \sum_{n\in \NA} \frac{\|\uc_n\|_{\CO^N}^2 }{| \lambda_n -\ve^2 z|^2}\right)^{1/2}  \leq C\,  \ve^{3/2}  \|\psi^{in}\|_{\HH^{in,\ve}},
\end{aligned}\]
because $\|\uc_n\|_{\CO^N}^2\leq C $ and $  \sum_{n\in \NA} | \lambda_n -\ve^2 z|^{-2}\leq  C \sum_{n\in \NA} | \lambda_n |^{-2} \leq C$. This proves the first  Claim in Eq.  \eqref{buried1}; the second  one  is trivial, being $G^{in,\ve}_z$ the adjoint of $\breve G^{in,\ve}_{\bar z}$. 
\end{proof}

\begin{Proposition}\label{p:Mgen} Let $z\in\CO\backslash\RE$. In the Generic Case,  
\begin{equation}\label{buried1.1}
\MM^{in,\ve}_z =\OO_{\BB(\CO^N)}(\ve).
\end{equation}
\end{Proposition}
\begin{proof}
Recall Eq.  \eqref{Min} and note that for any $\uv\in\CO^N$,  
\[
\| \MM^{in,\ve}_z \uv\|_{\CO^N } 
\leq \ve \sum_{n\in \NA} \frac{\|\uc_n\|_{\CO^N} \, |(\uc_n, \uv)_{\CO^N}|}{|\lambda_n -\ve^2 z|}  \leq \ve \|\uv\|_{\CO^N} \sum_{n\in \NA}\frac{\|\uc_n\|_{\CO^N}^2}{|\lambda_n -\ve^2 z|}  \leq C \ve \|\uv\|_{\CO^N},
\]
because $\|\uc_n\|_{\CO^N}^2\leq C $ and $  \sum_{n\in \NA} | \lambda_n -\ve^2 z|^{-1}\leq  C \sum_{n\in \NA} | \lambda_n |^{-1} \leq C$.
\end{proof}

\begin{Theorem}\label{t:mainG2}Let $z\in\CO\backslash\RE$. In the Generic Case
\begin{equation}\label{expgen}
R^\ve_z = 
\begin{pmatrix}
\mathring R^{out}_z  + \OO_{\BB(\HH^{out})}(\ve) & \OO_{\BB(\HH^{in,\ve},\HH^{out)}}(\ve^{3/2}) \\
\OO_{\BB(\HH^{out},\HH^{in,\ve})}(\ve^{3/2})& \OO_{\BB(\HH^{in,\ve})}(\ve^2) 
\end{pmatrix},
\end{equation}
where the expansion has to be understood in the  $out/in$ decomposition \eqref{outindec}. 
\end{Theorem}
\begin{proof}
Note that $\big( \MM^{in,\ve}_z\MM^{out}_z- \ID_N\big)^{-1} = \OO_{\CO^N}(1)$ by Eq. \eqref{buried1.1} and because $\MM^{out}_z$ is bounded and does not depend on $\ve$. Hence,  $\big( \MM^{in,\ve}_z\MM^{out}_z- \ID_N\big)^{-1} \MM^{in,\ve}_z = \OO_{\CO^N}(\ve)$.\\ 
 To conclude,  by Eq.s \eqref{RR1} - \eqref{RR4}, and by expansions \eqref{buried1}, we infer: $\RR_z^{out,out,\ve} = \OO_{\BB(\HH^{out})}(\ve)$; $\RR_z^{out,in,\ve} =  \OO_{\BB(\HH^{in,\ve},\HH^{out)}}(\ve^{3/2})$;  $\RR_z^{in,out,\ve} = \OO_{\BB(\HH^{out},\HH^{in,\ve})}(\ve^{3/2})$  (this is obvious since it is the adjoint of  $\RR_{\bar z}^{out,in,\ve} $); and $\RR_z^{in,in,\ve}    =  \OO_{\BB(\HH^{in,\ve})}(\ve^3) $. 
 
 Expansion \eqref{expgen} follows by taking into account the bound \eqref{buried0}, and from Rem. \ref{r:RR}. 
\end{proof}

Th. \ref{t:mainG} is a direct consequence of Th. \ref{t:mainG2}. 

\subsection{Non-Generic Case. Proof of Th. \ref{t:mainNG}.} In this section we study the limit of the relevant quantities in the Non-Generic Case  and prove Th. \ref{t:mainNG}.

Recall that, in the Non-Generic Case, $\{\hat \varphi_k^\ve\}_{k=1,\dots,m}$ denotes a set of orthonormal eigenfunctions corresponding to the zero eigenvalue, see also Rem. \ref{r:4.2}. 

\begin{Proposition}Let $z\in\CO\backslash\RE$.  In the Non-Generic Case
\begin{align}
\mathring R_z^{in,\ve} = & -\sum_{k=1}^m   \frac{\hat \varphi_k^\ve(\hat \varphi_k^\ve, \cdot)_{\HH^{in,\ve}} }{z}  + \OO_{\BB(\HH^{in,\ve})}(\ve^2); \label{buried3.0} \\ 
\breve G^{in,\ve}_z   = &-\sum_{k=1}^m \frac{\uch_k (\hat \varphi_k^\ve, \cdot)_{\HH^{in,\ve}}}{\ve^{1/2} z} +  \OO_{\BB(\HH^{in,\ve},\CO^N)} (\ve^{3/2}); \label{buried3}
\\
\qquad
G^{in,\ve}_z = & -\sum_{k=1}^m\frac{ \hat \varphi_k^\ve(\uch_k, \cdot)_{\CO^N}}{\ve^{1/2} z}  + \OO_{\BB(\CO^N, \HH^{in,\ve})}(\ve^{3/2}).
\label{buried4}
\end{align}
\end{Proposition}
\begin{proof}
We prove first Claim \eqref{buried3.0}. By Eq. \eqref{Rinve} we infer 
\begin{equation}\label{rat}
\mathring R^{in,\ve}_z = - \sum_{k=1}^m \frac{\hat \varphi_k^\ve ( \hat \varphi_k^\ve,\cdot)_{\HH^{in,\ve}}}{z} + 
\ve^2 \sum_{n:\lambda_n\neq0} \frac{\varphi_n^\ve ( \varphi_n^\ve,\cdot)_{\HH^{in,\ve}}}{\lambda_n -\ve^2 z}.
\end{equation}
Note that the second sum runs over  $\lambda_n\neq0$, hence one has the bound $|\la_n -\ve^2 z|\geq |\la_n|/2\geq C$, for $\ve$ small enough. For this reason, the bound in Eq. \eqref{buried3.0} on the second term at the r.h.s. of Eq. \eqref{rat} can be obtained with an argument similar to the one used in the proof of bound \eqref{buried0}. 

To prove Claim \eqref{buried3} we proceed in a similar way. We note that, see Eq. \eqref{GGbrevein},  
\[
\breve G^{in,\ve}_z = -\sum_{k=1}^m \frac{\uch_k (\hat \varphi_k^\ve, \cdot)_{\HH^{in,\ve}}}{\ve^{1/2} z}  + \ve^{3/2} \sum_{n:\la_n\neq 0 } \frac{\uc_n (\varphi_n^\ve, \cdot)_{\HH^{in,\ve}}}{\la_n-\ve^2 z},
\]
and bound the second term at the r.h.s.  by reasoning in the same way as in the proof of Prop. \ref{p:egg}. Claim \eqref{buried4} follows by noticing that $G^{in,\ve}_z$ is the adjoint of  $ \breve G^{in,\ve}_{\bar z}$.
\end{proof}

Next we prove a proposition on the expansion of the $N\times N$, $z$-dependent matrix $\MM_z^{in,\ve}$. Recall that $\WHC$ was defined in Def. 
\ref{d:WHP}.
%The expansion involves the operators $\WHC$ and $\WHP$, for their definitions we refer to Def. \ref{d:WHP}. Recall moreover that $\WHP  \MM^{out}_z \WHP$ is invertible by  Lemma \ref{l:Rout}. 
\begin{Proposition}Let $z\in\CO\backslash\RE$. In the Non-Generic Case, 
\begin{equation}\label{inverse0}
\MM_z^{in,\ve}= -  \frac{1}{\ve z }\WHC+  \OO_{\BB(\CO^N)}(\ve).  
\end{equation}
\end{Proposition}
\begin{proof}
The claim  immediately follows  from Eq. \eqref{Min}, after noticing that 
\[
\MM_z^{in,\ve} =  -  \frac{1}{ \ve z}\WHC  + 
  \ve \sum_{n:\lambda_n\neq 0 } \frac{ \uc_n (\uc_n, \cdot)_{\CO^N}}{\la_n- \ve^2 z}  
\]
and by treating the second term at the r.h.s. with argument similar to the one used in the proof of Prop. \ref{p:Mgen}. 
\end{proof}

We set 
\[\widetilde \MM_z^{in,\ve} := \ve  \MM_z^{in,\ve}
\]
and recall that  $\MM_z^{out}$ is invertible (see Rem. \ref{r:invertibility}), then 
\begin{equation}\label{bone}
 (\MM_z^{in,\ve}\MM^{out}_z -\ID_N)^{-1} =  \ve  {\MM^{out}_z}^{-1}(\widetilde \MM_z^{in,\ve}-\ve {\MM^{out}_z}^{-1})^{-1} . 
\end{equation}
In the following proposition we give an expansion formula for the term $(\widetilde \MM_z^{in,\ve}-\ve {\MM^{out}_z}^{-1})^{-1}$ in the Non-Generic Case.

\begin{Proposition}\label{p:Nz} Let $z\in\CO\backslash\RE$. In the Non-Generic Case, decompose the space $\CO^N$ as $\CO^N = \WHP\CO^N\oplus \WHP^\perp\CO^N$, and denote by $\WHC_0$ the restriction of $\WHC$ to $ \WHP\CO^N$. Then, the map   $\WHP^\perp{\MM^{out}_z}^{-1}\WHP^\perp$ is invertible in $\WHP^\perp\CO^N$.

Set 
 \begin{equation}\label{Nz}
N_z :=  (\WHP^\perp{\MM^{out}_z}^{-1}\WHP^\perp )^{-1}: \WHP^\perp\CO^N \to \WHP^\perp \CO^N,
\end{equation}
then 
\begin{equation}\label{pig}
\begin{aligned}
& (\widetilde \MM_z^{in,\ve}-\ve {\MM^{out}_z}^{-1})^{-1} \\ 
= &-  \begin{pmatrix}
z\WHC_0^{-1}+ \OO_{\BB(\WHP\CO^N)}(\ve )& - z\WHC_0^{-1} \WHP{\MM^{out}_z}^{-1}\WHP^\perp N_z+  \OO_{\BB(\WHP^\perp\CO^N,\WHP\CO^N)}(\ve)\\
 - z N_z\WHP^\perp {\MM^{out}_z}^{-1}\WHP \WHC_0^{-1} + \OO_{\BB(\WHP\CO^N,\WHP^\perp\CO^N)}(\ve) &  \ve^{-1}N_z +\OO_{\BB(\WHP^\perp\CO^N)}(1)
\end{pmatrix} , 
\end{aligned}
\end{equation} 
to be understood in the decomposition  $\CO^N = \WHP\CO^N\oplus \WHP^\perp\CO^N$. 
\end{Proposition}
\begin{proof}
We postpone the proof of the fact that the map $\WHP^\perp{\MM^{out}_z}^{-1}\WHP^\perp$  is invertible in $\WHP^\perp\CO^N$ to the appendix, see Rem. \ref{r:Minv}. 

Next we prove that the expansion formula \eqref{pig} holds true. We start by noticing that the map $z^{-1}\WHC+\ve {\MM^{out}_z}^{-1}$ is   invertible. In fact, by Rem. \ref{r:Minv}  and  
 since $(\uv,\WHC \uv)_{\CO^N} = \sum_{k=1}^m |(\uch_k,\uv)_{\CO^N}|^2 \geq0$,  we infer 
\[ 
\Im  \big(\uv,(z^{-1}\WHC+\ve {\MM^{out}_z}^{-1})\uv\big)_{\CO^N} =- \frac{\Im z}{|z|^2} (\uv,\WHC \uv)_{\CO^N}  -\ve\Im z \|G_z^{out} {\MM^{out}_z}^{-1} \uv \|^2_{\HH^{out}} \neq 0 ,
\]
because it is the sum of two  non-positive (or non-negative) terms and $\|G_z^{out} {\MM^{out}_z}^{-1} \uv \|^2_{\HH^{out}}\neq 0$ by the injectivity of $G_z^{out} {\MM^{out}_z}^{-1}$, see Rem. \ref{r:injectivity}. 

Moreover we have the a-priori estimate 
\begin{equation}\label{5.17a}
(\widetilde \MM_z^{in,\ve}-\ve {\MM^{out}_z}^{-1})^{-1} = \OO_{\BB(\CO^N)}(\ve^{-1}).
\end{equation}
The latter follows from (see also Eq. \eqref{berry})
\[\begin{aligned}
\|\uv\|_{\CO^N}\| (\widetilde \MM_z^{in,\ve}-\ve {\MM^{out}_z}^{-1})\uv\|_{\CO^N}  \geq  &
| (\uv, \widetilde \MM^{in,\ve}_z-\ve {\MM^{out}_z}^{-1}\uv)_{\CO^N} | \\ 
\geq  &
|\Im  (\uv, \widetilde \MM^{in,\ve}_z-\ve {\MM^{out}_z}^{-1}\uv)_{\CO^N}| \\ 
 =& \ve |\Im z| (\|G_z^{in,\ve} \uv \|^2_{\HH^{in,\ve}}+ \|G_z^{out} {\MM^{out}_z}^{-1} \uv \|^2_{\HH^{out}})\geq\ve C_z \|\uv\|^2_{\CO^N}, 
\end{aligned}
\]
for some positive constant $C_z$, from the injectivity of $G_z^{out} {\MM^{out}_z}^{-1}$. Hence, setting $\uv = (\widetilde \MM_z^{in,\ve}-\ve {\MM^{out}_z}^{-1})^{-1}\uu$, it follows that  $\|(\widetilde \MM_z^{in,\ve}-\ve {\MM^{out}_z}^{-1})^{-1}\uu\|_{\CO^N}\leq (\ve C_z)^{-1}\|\uu\|_{\CO^N}$. 

Next we use the expansion  (see Eq. \eqref{inverse0})
\begin{equation} \label{expwtMM}
\widetilde \MM_z^{in,\ve}= -  \frac{1}{z }\WHC+  \OO_{\BB(\CO^N)}(\ve^2), 
\end{equation}
which, together with the a-priori estimate \eqref{5.17a},  gives 
\begin{align}
(\widetilde \MM_z^{in,\ve}-\ve {\MM^{out}_z}^{-1})^{-1} = & - (z^{-1}\WHC+\ve {\MM^{out}_z}^{-1})^{-1} + (z^{-1}\WHC+\ve {\MM^{out}_z}^{-1})^{-1}  \OO_{\BB(\CO^N)}(\ve^2) (\widetilde \MM_z^{in,\ve}-\ve {\MM^{out}_z}^{-1})^{-1}\nonumber \\ 
= & - (z^{-1}\WHC+\ve {\MM^{out}_z}^{-1})^{-1} + (z^{-1}\WHC+\ve {\MM^{out}_z}^{-1})^{-1}  \OO_{\BB(\CO^N)}(\ve). \label{beg1} 
\end{align}
Here we used the formula $(A+B)^{-1} = A^{-1}- A^{-1}B (A+B)^{-1}$. Note that by using instead the complementary formula $(A+B)^{-1} = A^{-1}- (A+B)^{-1}B A^{-1}$, we obtain 
\begin{equation}\label{beg2}
(\widetilde \MM_z^{in,\ve}-\ve {\MM^{out}_z}^{-1})^{-1} = - (z^{-1}\WHC+\ve {\MM^{out}_z}^{-1})^{-1} +\OO_{\BB(\CO^N)}(\ve) (z^{-1}\WHC+\ve    {\MM^{out}_z}^{-1})^{-1}. 
\end{equation}

Next we analyze the term $(z^{-1}\WHC+\ve {\MM^{out}_z}^{-1})^{-1}$.
%In the   decomposition  $\CO^N = \WHP\CO^N\oplus \WHP^\perp\CO^N$ 
%\[
%\WHC = \begin{pmatrix}
%\WHC_0 & \NULL \\ 
%\NULL & \NULL
%\end{pmatrix}
%\]
%and  the map $\WHC_0$  is invertible  in $\WHP\CO^N$ (see also  Rem. \ref{r:WHP}). 

We start by noticing that the map $z^{-1}\WHC_0 +\ve \WHP{\MM^{out}_z}^{-1} \WHP: \WHP \CO^N\to \WHP\CO^N $ is invertible, because  $\WHC_0$ is invertible in $\WHP\CO^N$ and $\ve \WHP{\MM^{out}_z}^{-1} \WHP = \OO_{\CO^N}(\ve)$. 

By the  identification  (to be understood in the decomposition $\CO^N = \WHP\CO^N\oplus\WHP^\perp\CO^N$)
\begin{equation}\label{fugitive}
 {\MM^{out}_z}^{-1} =\begin{pmatrix}
\WHP {\MM^{out}_z}^{-1} \WHP & \WHP {\MM^{out}_z}^{-1} \WHP^\perp \\
\WHP^\perp {\MM^{out}_z}^{-1} \WHP &\WHP^\perp {\MM^{out}_z}^{-1} \WHP^\perp
\end{pmatrix},
\end{equation}
we have the identity 
\[
 z^{-1}\WHC+\ve {\MM^{out}_z}^{-1}
 = 
\begin{pmatrix}
z^{-1}\WHC_0 +\ve \WHP{\MM^{out}_z}^{-1}\WHP &  \ve \WHP{\MM^{out}_z}^{-1}\WHP^\perp \\
\ve \WHP^\perp {\MM^{out}_z}^{-1}\WHP &  \ve \WHP^\perp{\MM^{out}_z}^{-1}\WHP^\perp 
\end{pmatrix}.
\]
Hence, from the block-matrix inversion formula, we obtain 
\[
(z^{-1}\WHC+\ve {\MM^{out}_z}^{-1})^{-1} 
=  \begin{pmatrix}
D^\ve_z & - D_z^\ve \WHP{\MM^{out}_z}^{-1}\WHP^\perp N_z \\
 - N_z\WHP^\perp {\MM^{out}_z}^{-1}\WHP D_z^\ve
 &  \ve^{-1}N_z +N_z\WHP^\perp {\MM^{out}_z}^{-1}\WHP D_z^{\ve} \WHP{\MM^{out}_z}^{-1}\WHP^\perp N_z
\end{pmatrix},
\]
with  $D^\ve_z: \WHP\CO^N \to \WHP \CO^N $ given by 
\[
D^\ve_z := \Big(z^{-1}\WHC_0 +\ve \WHP{\MM^{out}_z}^{-1}\WHP - 
\ve \WHP{\MM^{out}_z}^{-1}\WHP^\perp ( \WHP^\perp{\MM^{out}_z}^{-1}\WHP^\perp )^{-1}\WHP^\perp {\MM^{out}_z}^{-1}\WHP  \Big)^{-1};
\]
note that $D^\ve_z$ is well-defined because it is the inverse of a map of the form $z^{-1}\WHC_0 + \OO_{\BB(\WHP\CO^N)}(\ve)$, and $ z^{-1}\WHC_0$ is invertible in $\WHP\CO^N$. 
Moreover, it holds true, 
\[
D^\ve_z = z\WHC_0^{-1} + \OO_{B(\WHP\CO^N)}(\ve).
\]
Hence, 
\[\begin{aligned}
&(z^{-1}\WHC+\ve {\MM^{out}_z}^{-1})^{-1}  \\ 
= &  \begin{pmatrix}
z\WHC_0^{-1}& - z\WHC_0^{-1} \WHP{\MM^{out}_z}^{-1}\WHP^\perp N_z \\
 - z N_z\WHP^\perp {\MM^{out}_z}^{-1}\WHP \WHC_0^{-1}
 &  \ve^{-1}N_z +z N_z\WHP^\perp {\MM^{out}_z}^{-1}\WHP \WHC_0^{-1} \WHP{\MM^{out}_z}^{-1}\WHP^\perp N_z
\end{pmatrix} + \OO_{\BB(\CO^N)}(\ve). 
\end{aligned}
\]
The latter can also be written as 
\[
(z^{-1}\WHC+\ve {\MM^{out}_z}^{-1})^{-1} = \begin{pmatrix}
z\WHC_0^{-1}& - z\WHC_0^{-1} \WHP{\MM^{out}_z}^{-1}\WHP^\perp N_z \\
 - z N_z\WHP^\perp {\MM^{out}_z}^{-1}\WHP \WHC_0^{-1}
 &  \ve^{-1}N_z +\OO_{\BB(\WHP^\perp\CO^N)}(1)
\end{pmatrix} + \OO_{\BB(\CO^N)}(\ve). 
\]
Using this expansion formula in Eq. \eqref{beg1}  we obtain 
\[\begin{aligned}
&(\widetilde \MM_z^{in,\ve}-\ve {\MM^{out}_z}^{-1})^{-1} \\ 
= &
-  \begin{pmatrix}
z\WHC_0^{-1}& - z\WHC_0^{-1} \WHP{\MM^{out}_z}^{-1}\WHP^\perp N_z \\
 - z N_z\WHP^\perp {\MM^{out}_z}^{-1}\WHP \WHC_0^{-1}
 &  \ve^{-1}N_z +\OO_{\BB(\WHP^\perp\CO^N)}(1)
\end{pmatrix} \\ 
& +  \begin{pmatrix}
z\WHC_0^{-1}& - z\WHC_0^{-1} \WHP{\MM^{out}_z}^{-1}\WHP^\perp N_z \\
 - z N_z\WHP^\perp {\MM^{out}_z}^{-1}\WHP \WHC_0^{-1}
 &  \ve^{-1}N_z +\OO_{\BB(\WHP^\perp\CO^N)}(1)
\end{pmatrix}\OO_{\BB(\CO^N)}(\ve) + \OO_{\BB(\CO^N)}(\ve)   \\ 
=& -  \begin{pmatrix}
z\WHC_0^{-1}+ \OO_{\BB(\WHP\CO^N)}(\ve )& - z\WHC_0^{-1} \WHP{\MM^{out}_z}^{-1}\WHP^\perp N_z+  \OO_{\BB(\WHP^\perp\CO^N,\WHP\CO^N)}(\ve)\\
 \OO_{\BB(\WHP\CO^N,\WHP^\perp\CO^N)}(1)
 &  \ve^{-1}N_z +\OO_{\BB(\WHP^\perp\CO^N)}(1)
\end{pmatrix} .
\end{aligned}
\] 
On the other hand, using  Eq. \eqref{beg2}, we obtain 
\[\begin{aligned}
(\widetilde \MM_z^{in,\ve}-\ve {\MM^{out}_z}^{-1})^{-1} = &
-  \begin{pmatrix}
z\WHC_0^{-1}& - z\WHC_0^{-1} \WHP{\MM^{out}_z}^{-1}\WHP^\perp N_z \\
 - z N_z\WHP^\perp {\MM^{out}_z}^{-1}\WHP \WHC_0^{-1}
 &  \ve^{-1}N_z +\OO_{\BB(\WHP^\perp\CO^N)}(1)
\end{pmatrix} \\ 
& + \OO_{\BB(\CO^N)}(\ve) \begin{pmatrix}
z\WHC_0^{-1}& - z\WHC_0^{-1} \WHP{\MM^{out}_z}^{-1}\WHP^\perp N_z \\
 - z N_z\WHP^\perp {\MM^{out}_z}^{-1}\WHP \WHC_0^{-1}
 &  \ve^{-1}N_z +\OO_{\BB(\WHP^\perp\CO^N)}(1)
\end{pmatrix} + \OO_{\BB(\CO^N)}(\ve)   \\ 
=& -  \begin{pmatrix}
z\WHC_0^{-1}+ \OO_{\BB(\WHP\CO^N)}(\ve )& \OO_{\BB(\WHP^\perp\CO^N,\WHP\CO^N)}(1)\\
 - z N_z\WHP^\perp {\MM^{out}_z}^{-1}\WHP \WHC_0^{-1} + \OO_{\BB(\WHP\CO^N,\WHP^\perp\CO^N)}(\ve)
 &  \ve^{-1}N_z +\OO_{\BB(\WHP^\perp\CO^N)}(1)
\end{pmatrix} .
\end{aligned}
\] 
Hence Expansion \eqref{pig} must hold true 
\end{proof}

Recall that, for  $\Im z\neq0$,  $\WHP M_z^{out}\WHP$ is invertible in $\WHP\CO^N$, see Rem. \ref{r:invertibility}. 
\begin{Proposition}\label{p:inverse} Let $z\in\CO\backslash\RE$. In the Non-Generic Case, 
\begin{equation*}\label{inverse1}
 (\MM_z^{in,\ve}{\MM^{out}_z}- \ID_N)^{-1} \MM_z^{in,\ve} 
 =\WHP (\WHP  \MM^{out}_z \WHP)^{-1}\WHP  + \OO_{\BB(\CO^N) }(\ve) .
\end{equation*}
\end{Proposition}
\begin{proof}
Taking into account   Expansion \eqref{expwtMM}, rewritten in the decomposition  $\CO^N = \WHP\CO^N\oplus \WHP^\perp\CO^N$, one has 
\[
\widetilde \MM_z^{in,\ve}= -  \frac{1}{ z }\WHC+  \OO_{\BB(\CO^N)}(\ve^2) = - \begin{pmatrix}
   z^{-1} \WHC_0 & 0 \\
   0 & 0  
\end{pmatrix}+\OO_{\BB(\CO^N) }(\ve^2).  
\]
So that, by Eq. \eqref{pig},  
\[(\widetilde \MM_z^{in,\ve}-\ve {\MM^{out}_z}^{-1})^{-1}\widetilde \MM_z^{in,\ve} =  \begin{pmatrix} \ID_{\WHP\CO^N} 
& 0 \\
 - N_z\WHP^\perp {\MM^{out}_z}^{-1}\WHP   & 0 
\end{pmatrix}+ \OO_{\BB(\CO^N) }(\ve). 
\]
By the latter expansion and by the  identification \eqref{fugitive} it follows that (recall Eq. \eqref{bone} and the definition of $N_z$ in Eq. \eqref{Nz})
\begin{align}
& (\MM_z^{in,\ve}{\MM^{out}_z}- \ID_N)^{-1} \MM_z^{in,\ve}  \nonumber \\ 
= & {\MM^{out}_z}^{-1} (\widetilde \MM_z^{in,\ve}-\ve {\MM^{out}_z}^{-1})^{-1}\widetilde \MM_z^{in,\ve} \nonumber \\ 
=&\begin{pmatrix}
\WHP {\MM^{out}_z}^{-1} \WHP & \WHP {\MM^{out}_z}^{-1} \WHP^\perp \nonumber \\
\WHP^\perp {\MM^{out}_z}^{-1} \WHP &\WHP^\perp {\MM^{out}_z}^{-1} \WHP^\perp
\end{pmatrix}
  \begin{pmatrix}\ID_{\WHP\CO^N} 
& 0 \\
 - N_z\WHP^\perp {\MM^{out}_z}^{-1}\WHP   & 0 
\end{pmatrix}+ \OO_{\BB(\CO^N) }(\ve) \nonumber \\ 
= & \begin{pmatrix}
\WHP {\MM^{out}_z}^{-1} \WHP -  \WHP {\MM^{out}_z}^{-1} \WHP^\perp  N_z\WHP^\perp {\MM^{out}_z}^{-1}\WHP  &0\\
0 &0 
\end{pmatrix}
  + \OO_{\BB(\CO^N) }(\ve). 
  \label{x1}
\end{align}
To conclude, we apply  the block-matrix inversion  formula to Eq. \eqref{fugitive} to obtain
\[
 \MM^{out}_z =\begin{pmatrix} \widetilde D_z 
& -\widetilde D_z  \WHP {\MM^{out}_z}^{-1} \WHP^\perp N_z\\ 
- N_z \WHP^\perp {\MM^{out}_z}^{-1} \WHP \widetilde D_z   & N_z + N_z \WHP^\perp  {\MM^{out}_z}^{-1} \WHP \widetilde D_z \WHP  {\MM^{out}_z}^{-1} \WHP^\perp N_z
\end{pmatrix},\]
with 
\[\widetilde D_z = (\WHP  {\MM^{out}_z}^{-1} \WHP - \WHP  {\MM^{out}_z}^{-1} \WHP^\perp N_z \WHP^\perp  {\MM^{out}_z}^{-1}\WHP)^{-1} . \]
Hence it must be  
\[
\WHP  \MM^{out}_z \WHP  = \widetilde D_z =  (\WHP  {\MM^{out}_z}^{-1} \WHP - \WHP  {\MM^{out}_z}^{-1} \WHP^\perp N_z \WHP^\perp  {\MM^{out}_z}^{-1}\WHP)^{-1}, \]
so that  
\[
(\WHP  \MM^{out}_z \WHP)^{-1} =\WHP  {\MM^{out}_z}^{-1} \WHP - \WHP  {\MM^{out}_z}^{-1} \WHP^\perp N_z \WHP^\perp  {\MM^{out}_z}^{-1}\WHP. \]
Which, together with Eq. \eqref{x1}, allows us to infer  the expansion 
\[
 (\MM_z^{in,\ve}{\MM^{out}_z}- \ID_N)^{-1} \MM_z^{in,\ve} 
 = \begin{pmatrix}
(\WHP  \MM^{out}_z \WHP)^{-1}  &0\\
0 &0 
\end{pmatrix}
  + \OO_{\BB(\CO^N) }(\ve)   = \WHP(\WHP  \MM^{out}_z \WHP)^{-1}  \WHP+ \OO_{\BB(\CO^N) }(\ve)\]
and conclude the proof of the proposition. 
\end{proof}

We are now ready to state and prove the main theorem for the Non-Generic Case. In the statement of the theorem, we assume that  $\Ker\WHC\subset \CO^N$, i.e., $\WHP\neq 0$. In this way the quantity $(\uch_k, \WHC_0^{-1} \uch_{k'} )_{\CO^N}$ is certainly well defined. We discuss the case  $\Ker\WHC=  \CO^N$ (i.e., $\WHP=0$) separately  in the proof of point (ii) of Th. \ref{t:mainNG} (after the proof of Th. \ref{t:mainNG2}). 

\begin{Theorem}\label{t:mainNG2}Let $z\in\CO\backslash\RE$. In the Non-Generic Case assume that  $\Ker\WHC\subset \CO^N$, then 
\begin{equation*}\label{expgen1}
R^\ve_z = 
\begin{pmatrix}
\WHR^{out}_z  + \OO_{\BB(\HH^{out})}(\ve) & \OO_{\BB(\HH^{in,\ve},\HH^{out)}}(\ve^{1/2}) \\
\OO_{\BB(\HH^{out},\HH^{in,\ve})}(\ve^{1/2})& - z^{-1}\sum_{k,k' = 1}^m \left(\delta_{k,k'} -(\uch_k, \WHC_0^{-1} \uch_{k'} )_{\CO^N} \right) \hat \varphi_k^\ve(\hat \varphi_{k'}^\ve, \cdot)_{\HH^{in,\ve}} + \OO_{\BB(\HH^{in,\ve}) }(\ve).  
\end{pmatrix},
\end{equation*}
where the expansion has to be understood in the  $out/in$ decomposition \eqref{outindec}. 
\end{Theorem}
\begin{proof}
We analyze term by term the r.h.s. in Eq. \eqref{RR0}. 

Term $out/out$. By Prop. \ref{p:inverse} and Lemma \ref{l:Rout}, it immediately follows that 
\[
\mathring R^{out}_z - \RR_z^{out,out,\ve} = \WHR^{out}_z  + \OO_{\BB(\HH^{out})}(\ve).  \]

Term $out/in$. By Eq. \eqref{bone} and by the definition of $ \RR_z^{out,in,\ve}$, recalling that   $G^{out}_z$ and ${\MM^{out}_z}^{-1}$ are bounded, it is enough to prove that 
\begin{equation} 
  \ve(\widetilde \MM_z^{in,\ve}-\ve {\MM^{out}_z}^{-1})^{-1} \breve G^{in,\ve}_z= \OO_{\BB({\HH^{in,\ve}},\CO^N)} (\ve^{1/2}).\label{OD_NG1} \end{equation}
Taking into account the fact that for all $\psi\in\HH^{in,\ve}$,  $\|\sum_{k=1}^m \uch_k (\hat \varphi_k^\ve, \psi )_{\HH^{in,\ve}}\|_{\CO^N} \leq C \|\psi\|_{\HH^{in,\ve}}$, and the fact that $\sum_{k=1}^m \uch_k (\hat \varphi_k^\ve, \psi )_{\HH^{in,\ve}} \in \WHP\CO^N$ (it is a linear combination of vectors in $\WHP\CO^N$, see Rem \ref{r:uch}) we infer that (see Eq. \eqref{buried3}), 
\[
\breve G_z^{in,\ve}\psi = \uv^\ve + \uu^\ve \qquad \uv^\ve : = -\sum_{k=1}^m \frac{\uch_k (\hat \varphi_k^\ve, \psi)_{\HH^{in,\ve}}}{\ve^{1/2} z}
\]
with $\uv^\ve\in \WHP \CO^N$,  $\|\uv^\ve\|_{\CO^N}\leq C \ve^{-1/2} \|\psi\|_{\HH^{in,\ve}}$, and $\|\uu^\ve\|_{\CO^N}\leq C \ve^{3/2} \|\psi\|_{\HH^{in,\ve}}$. 

Hence,  by the expansion \eqref{pig}, we infer 
\begin{equation}\label{notte}
\begin{aligned}
& \ve (\widetilde \MM_z^{in,\ve}-\ve {\MM^{out}_z}^{-1})^{-1} \breve G_z^{in,\ve}\psi \\ 
= & 
-\ve\big(z\WHC_0^{-1}- z N_z\WHP^\perp {\MM^{out}_z}^{-1}\WHP \WHC_0^{-1} + \OO_{\BB(\CO^N)}(\ve ) \big)\uv^\ve
+\ve (\widetilde \MM_z^{in,\ve}-\ve {\MM^{out}_z}^{-1})^{-1} \uu^\ve.
\end{aligned}
\end{equation}
Here the leading term is 
\[
\ve\big(z\WHC_0^{-1}- z N_z\WHP^\perp {\MM^{out}_z}^{-1}\WHP \WHC_0^{-1}\big)\uv^\ve, 
\]
and for it we have the bound 
\[
\|\ve\big(z\WHC_0^{-1}- z N_z\WHP^\perp {\MM^{out}_z}^{-1}\WHP \WHC_0^{-1}\big)\uv^\ve\|_{\CO^N}\leq
 C \ve^{1/2} \|\psi\|_{\HH^{in,\ve}}. \]
The remainder is bounded by 
\[
\|\OO_{\BB(\CO^N)}(\ve^2 )\uv^\ve
+\ve (\widetilde \MM_z^{in,\ve}-\ve {\MM^{out}_z}^{-1})^{-1} \uu^\ve
\|_{\CO^N}\leq C\ve^2 \|\uv^\ve\|_{\CO^N} + C\|\uu^\ve\|_{\CO^N} \leq 
 C \ve^{3/2} \|\psi\|_{\HH^{in,\ve}};\]
in the latter bound we  used $(\widetilde \MM_z^{in,\ve}-\ve {\MM^{out}_z}^{-1})^{-1} = \OO_{\BB(\CO^N)}(\ve^{-1})$, see Eq. \eqref{pig} (see also Eq. \eqref{5.17a}). Hence,
\[
\|\ve (\widetilde \MM_z^{in,\ve}-\ve {\MM^{out}_z}^{-1})^{-1} \breve G_z^{in,\ve}\psi \|_{\CO^N} \leq C \ve^{1/2} \|\psi\|_{\HH^{in,\ve}},
\]
and the bound \eqref{OD_NG1} holds true. 

The bound on the term $in/out$ follows immediately by noticing that $ \RR_z^{in,out,\ve} = \RR_{\bar z}^{out,in,\ve*}$.

Term $in/in$. By Eq. \eqref{bone}, we have that 
\[
\RR_z^{in,in,\ve}     = \ve G^{in,\ve}_z    (\widetilde \MM_z^{in,\ve}-\ve {\MM^{out}_z}^{-1})^{-1}\breve G^{in,\ve}_z. 
\]
Taking into account Eq. \eqref{notte} and the expansion \eqref{buried4}, we infer that, for all $\psi\in \HH^{in,\ve}$ the leading term in $\RR_z^{in,in,\ve} \psi$ is given by 
\[\begin{aligned}
\sum_{k=1}^m\frac{ \hat \varphi_k^\ve(\uch_k, \cdot)_{\CO^N}}{\ve^{1/2} z} \big(\ve\big(z\WHC_0^{-1}- z N_z\WHP^\perp {\MM^{out}_z}^{-1}\WHP \WHC_0^{-1}\big)\uv^\ve\big) = & \ve^{1/2}\sum_{k=1}^m \hat \varphi_k^\ve(\uch_k, \WHC_0^{-1}\uv^\ve)_{\CO^N} \\ 
= & - \frac1z \sum_{k,k'=1}^m \hat \varphi_k^\ve(\uch_k, \WHC_0^{-1} \uch_{k'} )_{\CO^N}(\hat \varphi_{k'}^\ve, \psi)_{\HH^{in,\ve}}. 
\end{aligned}\]
the remainder being of order $\ve$. From the latter formula and from the expansion \eqref{buried3.0} we infer 
\[
\mathring R^{in,\ve}_z-\RR_z^{in,in,\ve}
=  - z^{-1}\sum_{k,k' = 1}^m \left(\delta_{k,k'} -(\uch_k, \WHC_0^{-1} \uch_{k'} )_{\CO^N} \right) \hat \varphi_k^\ve(\hat \varphi_{k'}^\ve, \cdot)_{\HH^{in,\ve}} + \OO_{\BB(\HH^{in,\ve}) }(\ve). 
\]
\end{proof}
Th. \ref{t:mainNG} - (i) follows immediately from Th. \ref{t:mainNG2}.
\begin{proof}[{\bf Proof of Th. \ref{t:mainNG} - (ii)}] If $\Ker \WHC = \CO^N$ then $\uch_k=0$, for all $k = 1,\dots,m$, see Rem. \ref{r:uch}. Hence, expansions \eqref{buried3}, \eqref{buried4}, and \eqref{inverse0} read respectively
\[
\breve G^{in,\ve}_z   =  \OO_{\BB(\HH^{in,\ve},\CO^N)} (\ve^{3/2}); 
\qquad
G^{in,\ve}_z =  \OO_{\BB(\CO^N, \HH^{in,\ve})}(\ve^{3/2});\qquad 
\MM_z^{in,\ve} = \OO_{\BB(\CO^N)}(\ve).
\]
Reasoning along the lines of the analysis of the Generic Case, see the proof of Th. \ref{t:mainG2}, and taking into account the expansion \eqref{buried3.0}, one readily infers 
\begin{equation*}
R^\ve_z = 
\begin{pmatrix}
\mathring R^{out}_z  + \OO_{\BB(\HH^{out})}(\ve) & \OO_{\BB(\HH^{in,\ve},\HH^{out)}}(\ve^{3/2}) \\
\OO_{\BB(\HH^{out},\HH^{in,\ve})}(\ve^{3/2})& -\sum_{k=1}^m   \frac{\hat \varphi_k^\ve(\hat \varphi_k^\ve, \cdot)_{\HH^{in,\ve}} }{z}  + \OO_{\BB(\HH^{in,\ve})}(\ve^2),
\end{pmatrix}
\end{equation*}
which implies the statement in Th. \ref{t:mainNG} - (ii). 
\end{proof}
\begin{proof}[{\bf Proof of Th. \ref{t:mainNG} - (iii)}]
 To prove the second part of Th. \ref{t:mainNG},  recall that $\uch_{k'} \in \WHP\CO^N$ and $\WHC_0^{-1}\uch_{k'} \in \WHP\CO^N$, hence $\WHC \WHC_0^{-1} \uch_{k'} =\WHC_0 \WHC_0^{-1}\uch_{k'} = \uch_{k'}$. By the definition of $\WHC$ this is equivalent to 
 \[
 \sum_{k=1}^m (\delta_{k,k'} - (\uch_k, \WHC_0^{-1}\uch_{k'})) \uch_k= 0.  
 \]
 If the vectors $\{\uch_k\}_{k=1}^m$ are linearly independent this linear combination is zero if and only if $\delta_{k,k'} - (\uch_k, \WHC_0^{-1}\uch_{k'}) = 0$ for all $k$. Hence,  expansion \eqref{nona1} follows  from Eq. \eqref{nona0}. 
\end{proof}

\begin{Remark} Denote by $\Lambda$ the operator in $\HH^{in,\ve}$ defined by 
\[
D(\Lambda) := \HH^{in,\ve},\;\qquad   \Lambda := \sum_{k,k' = 1}^m \left(\delta_{k,k'} -(\uch_k, \WHC_0^{-1} \uch_{k'} )_{\CO^N} \right) \hat \varphi_k^\ve(\hat \varphi_{k'}^\ve, \cdot)_{\HH^{in,\ve}}. 
\]
$\Lambda$ is selfadjoint and $\Lambda^2= \Lambda$. The first claim is obvious (recall that $\WHC_0$ is selfadjoint). To prove the  second claim, note that, since $(\hat \varphi_{l'}^\ve, \varphi_{k}^\ve)_{\HH^{in,\ve}} = \delta_{l',k}$,   
\[
\Lambda^2 = \sum_{l,k,k' = 1}^m \left(\delta_{l,k} -(\uch_l, \WHC_0^{-1} \uch_{k} )_{\CO^N} \right)\left(\delta_{k,k'} -(\uch_k, \WHC_0^{-1} \uch_{k'} )_{\CO^N} \right) \hat \varphi_{l}^\ve(\hat \varphi_{k'}^\ve, \cdot)_{\HH^{in,\ve}},
\]
but 
\[\begin{aligned}
& \sum_{k = 1}^m \left(\delta_{l,k} -(\uch_l, \WHC_0^{-1} \uch_{k} )_{\CO^N} \right)\left(\delta_{k,k'} -(\uch_k, \WHC_0^{-1} \uch_{k'} )_{\CO^N} \right)   \\ 
=  & 
\delta_{l,k'} - 2 (\uch_l, \WHC_0^{-1} \uch_{k'} )_{\CO^N}  + \sum_{k = 1}^m (\uch_l, \WHC_0^{-1} \uch_{k} )_{\CO^N} (\uch_k, \WHC_0^{-1} \uch_{k'} )_{\CO^N}    \\ 
= & \delta_{l,k'} - 2 (\uch_l, \WHC_0^{-1} \uch_{k'} )_{\CO^N}  + (\uch_l, \WHC_0^{-1}  \WHC\WHC_0^{-1} \uch_{k'} )_{\CO^N}  
=  \delta_{l,k'} -  (\uch_l, \WHC_0^{-1} \uch_{k'} )_{\CO^N}  , 
\end{aligned}\]
where we used the fact that $\WHC_0^{-1}  \WHC\WHC_0^{-1} = \WHC_0^{-1}  \WHC_0 \WHC_0^{-1} = \WHC_0^{-1}$. Hence, 
\[
\Lambda^2 = \sum_{l,k' = 1}^m\left( \delta_{l,k'} -  (\uch_l, \WHC_0^{-1} \uch_{k'} )_{\CO^N} \right) \hat \varphi_{l}^\ve(\hat \varphi_{k'}^\ve, \cdot)_{\HH^{in,\ve}} = \Lambda. 
\]
Hence, $\Lambda$ is an orthogonal projection in $\HH^{in,\ve}$. 
\end{Remark}

\appendix

\section{\label{a:krf}Proof of the Kre{\u\i}n resolvent formulae}

We use several known results from  the theory of self-adjoint extensions of symmetric operators. 

We follow, for the most,  the approach and the notation from the papers by A.  Posilicano \cite{posilicano_jfa01} and \cite{pos-om08}. Other approaches would be possible, such as the one based on the use of boundary triples, see, e.g.,  \cite{alb-pan-jpa05, bruining-geyler-pankrashkin:08, gorbachuk-gorbachuk, schmudgen}.

When no misunderstanding is possible, in this appendix we omit the suffixes $out$, $in$, and $\ve$. 

\subsection{Proofs of Lemmata \ref{l:Rout} and \ref{l:Rve}.\label{s:A1}} 
We denote by $\mathring \tau$ the restriction of the maps $\tau$ to the domain $D(\mathring H)$, by Eq.s. \eqref{tau-out} and \eqref{tau-in} we infer 
\[
\mathring \tau : D(\mathring H^{\ve}) \to \CO^{2N},\qquad \mathring \tau = \diag(\mathring \tau^{out}, \mathring \tau^{in});
\]
\[
\mathring \tau^{out} : D(\mathring H^{out}) \to \CO^{N},\qquad  \mathring \tau^{out}\psi :=\Psi'(\zero); 
\]
\[
\mathring \tau^{in} : D(\mathring H^{in,\ve}) \to \CO^{N},\qquad 
\mathring \tau^{in}\psi :=( \psi(v_1), ..., \psi(v_N))^T; 
\]
where in $\mathring \tau^{in}$  we used the definition of $\tau^{in}$ and the fact that functions in $D(\mathring H^{in,\ve})$ are continuous in the connecting vertices. 

\begin{Remark}\label{r:injectivity}
The map $\mathring \tau$ is surjective. Hence, the map $\breve G_z^\ve = \tau \mathring R_z^\ve = \mathring \tau \mathring R_z^\ve$ is also surjective as a map from $\HH^\ve \to \CO^{2N}$ (the operator $ \mathring R_z^\ve:\HH^\ve\to D(\mathring H^\ve)$ is obviously surjective). We conclude that $G_z^\ve = \breve G_{\bar z}^{\ve*}$ is an injective map from $\CO^{2N}\to \HH^\ve$ (it is the adjoint of a surjective map). A similar statement holds true also for the corresponding ``$out$'' and ``$in$'' operators. 
\end{Remark}

\begin{Remark}\label{r:A1}
We claim that for all $z \in \rho(\mathring H^\ve)$ and $\uv\in \CO^{2N}$ one has $G_z^\ve \uv \in \HH_2^\ve$ and 
\begin{equation}\label{defect}
(-\Delta +B^\ve -z )G_z^\ve\uv = 0,
\end{equation}  
and similar properties hold true for the ``$out$'' and ``$in$'' operators (here $\Delta$ denotes the maximal Laplacian in $\HH^\ve$, i.e., $D(\Delta) := \HH_2^\ve$, $\Delta \psi = \psi''$). \\
To prove that $G_z^\ve \uv \in \HH_2^\ve$ and that Eq. \eqref{defect} holds true we start by discussing the case $B^\ve = 0$. In such a case it is possible to obtain an explicit formula for the integral kernel of $\mathring R_{z,0}^\ve = \mathring R_{z,B^\ve =0}^\ve$, see, e.g., \cite[Lemma 4.2]{kostrykin-schrader-cm06}. By this  explicit formula  it is easily seen that the operator $G^\ve_{z,0} = G^\ve_{z,B^\ve =0}$ maps any vector $\uv \in \CO^{2N}$ in a function in $\HH_2^\ve$ and that $(-\Delta -z )G_{z,0}^\ve\uv = 0$. It is not needed to investigate the detailed properties of the boundary conditions in the vertices of $\GG^\ve$, it is enough to take into account the dependence on $x,y\in\GG^\ve$ of the integral kernel $\mathring R_{z,0}^\ve (x,y)$ (see also \cite[Examples 5.1 and 5.2]{pos-om08}). That the same is true for $B^\ve\neq 0$ follows immediately from the resolvent identity 
\[
\mathring R_z^\ve = \mathring R_{z,0}^\ve  - \mathring R_{z,0}^\ve B^\ve \mathring R_z^\ve,    
\]
which gives $\breve G_z^\ve = \breve G_{z,0}^\ve  - \breve G_{z,0}^\ve B^\ve \mathring R_z^\ve$ and $ G_z^\ve =G_{z,0}^\ve  - \mathring R_z^\ve   B^\ve G_{z,0}^\ve$. 
\end{Remark}

In consideration of  the remark  above, 	we infer that the maps ($N\times N$, $z$-dependent  matrices)  $M_z$ in Eq.s \eqref{Mout}, \eqref{Min0} and \eqref{M} are all well defined.  Moreover, by the resolvent identities 
\[
R_z - R_w = (z-w)  R_zR_w \qquad \textrm{and} \qquad R_z = R_{\bar z}^* 
\]
it follows that
\[
\breve G_z - \breve G_w = (z-w)  \breve G_z R_w,
\]
\[
G_z - G_w = (z-w)   R_w G_z,
\]
\begin{equation}\label{MM}
M_z - M_w = (z-w)   \breve G_w G_z \qquad \textrm{and} \qquad M_z = M_{\bar z}^*.
\end{equation}
Let us denote by $\Kcal$ the space $\CO^{2N}$ or $\CO^N$ depending on if we are reasoning with operators in $\HH^{\ve}$, $\HH^{out}$ or $\HH^{in,\ve}$. 
By Eq. \eqref{MM}, it follows that for any projection $P$ in $\Kcal$ and any self-adjoint operator $\Theta$ in $\Ran P$,  the map $M_z^{P,\Theta} := PM_zP - \Theta$ is invertible in $\Ran P$. To see that this is indeed the case, note that by Eq. \eqref{MM} one has 
\[
M_z^{P,\Theta} - M_w^{P,\Theta} = (z-w) P  \breve G_w G_zP \qquad \textrm{and} \qquad M_z^{P,\Theta} = M_{\bar z}^{P,\Theta*}.
\]
So that, for $\Im z\neq 0$ and for  all $\uv\in\Kcal$, such that $P\uv\neq 0$, it holds 
\begin{equation}\label{berry}
\Im (\uv,M_z^{P,\Theta}\uv)_{\Kcal } =\frac{1}{2i} \big(\uv, (M_z^{P,\Theta} - M_{\bar z}^{P,\Theta})\uv\big)_{\Kcal}  = \Im z \|G_z P \uv\|_{\HH}^2 \neq 0; 
\end{equation}
because $G_z$ is injective. Hence, $M_z^{P,\Theta}$ is invertible in $\Ran P$ for $\Im z\neq 0$.
%%%By using again the fact that $G_z$ is injective, and by a similar argument we infer 
%%%\[
%%%\|\uv\|_\Kcal \|M_z^{P,\Theta}\uv\|_{\Kcal} \geq |(\uv,M_z^{P,\Theta}\uv)_{\Kcal } | \geq 
%%%|\Im (\uv,M_z^{P,\Theta}\uv)_{\Kcal }| \geq |\Im z| \|G_z P \uv\|_{\Kcal}^2 \geq C_z \|P \uv\|_{\Kcal}^2,
%%%\]
%%%for some positive constant $C_z$. Hence, setting $\uv = {M_z^{P,\Theta}}^{-1} \uu $, for some $\uu\in P\Kcal$, we have that $P \uv = \uv$ and $ \|P \uv\|_{\Kcal} \leq C_z^{-1} \|\uu\|_\Kcal$. 

\begin{Remark}\label{r:invertibility}
By the discussion above, it follows that the maps $M^{out}_z :\CO^N\to \CO^N$, $\WHP \MM^{out}_z\WHP: \WHP\CO^{N} \to \WHP\CO^{N}$, and $(\MM^\ve_z- \OD) : \CO^{2N} \to \CO^{2N}$ are  invertible for all $\Im z\neq0$. 
\end{Remark}

By \cite[Th. 2.1]{pos-om08} (see also \cite[Th. 2.1]{posilicano_jfa01})  it follows that:  for any $z\in \CO\backslash\RE$ the operators $\WHR^{out}_z$ and $R_z^\ve$ are the resolvents of a self-adjoint extension of the symmetric operators $\mathring H^{out}\restriction_{\Ker \mathring\tau^{out}}$ and  $\mathring H^{\ve}\restriction_{\Ker \mathring\tau}$ respectively. 

We are left to prove that such self-adjoint extensions coincide with $\WHH^{out}$ and $H^\ve$ respectively. 

Let us focus attention on $R_z^\ve$ (similar considerations hold true for $\WHR_z^{out}$).  Since the self-adjoint operator associated to $R_z^\ve$ is an extension of $\mathring H^{\ve}\restriction_{\Ker \mathring\tau}$, to prove that $R_z^\ve$ is  the resolvent of $H^\ve$, we  just need to check that  in the connecting vertices  functions in $\Ran R_z^\ve$ satisfy the boundary conditions required by $D(H^\ve)$. The remaining boundary conditions are clearly satisfied because the map $\mathring \tau$ evaluates functions only in the connecting vertices. 

Define the maps:
\begin{equation*}\label{sigma-out}
\sigma^{out} : \HH_2^{out} \to \CO^N \qquad \sigma^{out}\psi := \Psi(\zero);
\end{equation*}
\begin{equation*}\label{sigma-in}
\begin{aligned}
&\sigma^{in} : \HH_2^{in,\ve} \to \CO^N \\ 
&\sigma^{in}\psi :=- \left( {\sqrt{d^{in}(v_1)}} (\II_{d^{in}(v_1)},\Psi'(v_1))_{\CO^{d^{in}(v_1)}}, ..., {\sqrt{d^{in}(v_N)}}(\II_{d^{in}(v_N)},\Psi'(v_N))_{\CO^{d^{in}(v_N)}}\right)^T;
\end{aligned}
\end{equation*}
and 
\begin{equation*}\label{sigma}
\sigma:  \HH_2^\ve = \HH_2^{out}\oplus \HH_2^{in,\ve}  \to \CO^{2N} \qquad \sigma  := \diag(\sigma^{out},\sigma^{in}). 
\end{equation*}

We recall the following formula which is obtained by integrating by parts 
\begin{equation}\label{bt-iden}
\begin{aligned}
\big(  (-\Delta +B^\ve  - \bar z) \phi, \psi \big)_{\HH^\ve} -  \big(   \phi, (-\Delta +B^\ve  - z) \psi\big)_{\HH^\ve}  = 
\sum_{v\in \VV} \Big[ (\Phi'(v) , \Psi(v) )_{\CO^{d(v)}} - (\Phi(v), \Psi'(v) )_{\CO^{d(v)}}  \Big] \\ 
\qquad \qquad \qquad \qquad \qquad \qquad \qquad \qquad  \forall \phi,\psi \in\HH_2^\ve. 
\end{aligned}
\end{equation}

Fix  $\chi \in\HH^\ve$ and let $\uv = \big(\MM^\ve_z-\OD\big)^{-1}\breve{G}^\ve_z \chi \in\CO^{2N} $ and $\psi = G^\ve_z \uv$. 

For all $\phi\in D(\mathring H^\ve)$  and $\psi$ as above,  the identity \eqref{bt-iden} gives
\begin{equation}\label{triple}
\big( \tau \phi,\uv\big)_{\CO^{2N}}
= 
\sum_{v\in \CC} \Big[ ( {\KK^{in}_v}^\perp {\Phi^{in}}'(v) ,  {\KK^{in}_v}^\perp \Psi^{in}(v))_{\CO^{d(v)}} - ( \KK^{in}_v {\Phi^{in}}(v) ,   \KK^{in}_v {\Psi^{in}}'(v)  )_{\CO^{d(v)}}  \Big]
+ \sum_{j=1}^N   {\overline{\phi^{out}_j}}'(0)   \psi_j^{out}(0).
\end{equation}

In what follows we use the decomposition $\CO^{2N} = \CO^N \oplus \CO^N$, so that $\uv = (\uv^{out},\uv^{in})$ and  $ \tau \phi = ( \tau^{out} \phi^{out}, \tau^{in} \phi^{in})$. 

Let $\phi = (\phi^{out},0)\in D(\mathring H^\ve)$. Then Identity \eqref{triple} gives  
\begin{equation}\label{tripleout}
\big( \tau^{out} \phi^{out},\uv^{out}\big)_{\CO^{N}}
= \sum_{j=1}^N   {\overline{\phi^{out}_j}}'(0)   \psi_j^{out}(0).
\end{equation}
Take $\phi^{out}\in D(\mathring H^{out})$, such that ${\phi_1^{out}}'(0)=1$ and $\phi^{out}_j = 0$ for all $j = 2,\dots,N$. Then $(\tau^{out} \phi^{out})_j= \delta_{1,j}$, $j=1,\dots,N$ and Eq. \eqref{tripleout} gives $\psi_1^{out}(0) = q_1$. In a similar way it is possible to show that $\psi_j^{out}(0) = q_j$ for all $j=2,\dots,N$. Hence, $\sigma^{out}\psi^{out} = \uv^{out}$. 

Next let $\phi = (0, \phi^{in})$. Then Identity \eqref{triple} gives 
\begin{equation}\label{triplein}
\big( \tau^{in} \phi^{in},\uv^{in}\big)_{\CO^{N}}
= 
\sum_{v\in \CC} \Big[ ( {\KK^{in}_v}^\perp {\Phi^{in}}'(v) ,  {\KK^{in}_v}^\perp \Psi^{in}(v))_{\CO^{d(v)}} - ( \KK^{in}_v {\Phi^{in}}(v) ,   \KK^{in}_v {\Psi^{in}}'(v)  )_{\CO^{d(v)}}  \Big].
\end{equation}
Take $\phi^{in}$ such that   $\phi^{in}(v_1) =1$, ${\Phi^{in}}'(v_1) =0 $ and ${\Phi^{in}}'(v_j) = \Phi^{in}(v_j) = 0 $ for all $j = 2,\dots, N$. Hence, $(\tau^{in}\phi^{in})_j = \delta_{1,j}$, $j=1,\dots,N$, and  $ \KK^{in}_{v_1} {\Phi^{in}}(v_1) = (d^{in}(v_1))^{1/2} \II_{d^{in}(v_1)}$. Hence, Eq. \eqref{triplein} gives 
\begin{equation*}
\begin{aligned}
q^{in}_1= 
 - \big( (d^{in}(v_1))^{1/2} \II_{d^{in}(v_1)}  ,   \KK^{in}_{v_1} {\Psi^{in}}'(v_1)  \big)_{\CO^{d(v_1)}}  =  
 - \big( (d^{in}(v_1))^{1/2} \II_{d^{in}(v_1)}  ,  {\Psi^{in}}'(v_1)  \big)_{\CO^{d(v_1)}} = (\sigma^{in} \psi^{in})_1.
\end{aligned}
\end{equation*}
In a similar way one can prove $q^{in}_j= (\sigma^{in} \psi^{in})_j$, $j=2,\dots,N$, hence,  $\sigma^{in} \psi^{in} = \uv^{in}$. 

We also note that the function $\psi$ is continuous in the connecting vertices (whenever the vertex degree is larger or equal than two).  To see that this is indeed the case, consider in Eq. \eqref{triplein} a function $\phi^{in}$ such that $\phi^{in}(v_j) = 0$, $j=1, \dots, N$, ${\Phi^{in}}'(v_1) = (1,-1,0,\dots,0)^T := \underline e$, ${\Phi^{in}}'(v_j)=0$, $j=2,\dots,N$. Since $\KK_{v_1}^{in\perp} \underline e = \underline e$, condition \eqref{triplein} gives $(\underline e, \Psi^{in}(v_1)) = 0$. Repeating the process, moving $-1$ in the vector $\underline e$ on all the positions (from the second one on) one obtains the continuity of $\psi$ in the vertex $v_1$. The same holds true for every connecting vertex. 

We have proved that for any   $\chi \in\HH^\ve$, setting  $\uv = \big(\MM^\ve_z-\OD\big)^{-1}\breve{G}^\ve_z \chi \in\CO^{2N}$,  one has: 
\begin{equation}\label{sigmaG}
\sigma^{out} G_z^{out} \uv^{out} = \uv^{out}\,; \qquad \sigma^{in} G_z^{in,\ve} \uv^{in} = \uv^{in}\,;\qquad  \sigma G_z^{\ve} \uv = \uv. 
\end{equation}

Let $\chi\in\HH^\ve$ and set $\psi =  R_z^\ve \chi $. One has that 
\[
\tau  \psi = \tau \big( \mathring{R}^\ve_z - G^\ve_z \big(\MM^\ve_z-\OD\big)^{-1}\breve{G}^\ve_z \big) \chi =\big(\ID - M_z^\ve (\MM^\ve_z-\OD)^{-1}\big) \breve{G}^\ve_z  \chi = - \OD  (\MM^\ve_z-\OD)^{-1}\breve{G}^\ve_z \chi. 
\]
On the other hand, noticing that $\sigma  \mathring{R}^\ve_z  \chi = 0$, by the definition of $D(\mathring H^\ve) $ (see Eq.s \eqref{DHout}, \eqref{DHinve}, and \eqref{DHve}), and by Eq. \eqref{sigmaG}  it follows that 
\[
\sigma \psi = - (\MM^\ve_z-\OD)^{-1}\breve{G}^\ve_z\chi. 
\]
We conclude that $\psi$ satisfies the condition $ \tau \psi = \OD \sigma \psi$. Taking into account the fact that $\psi^{in}$ is continuous in the connecting vertices, it is easy convince oneself that the condition $\tau \psi = \OD \sigma \psi$ is equivalent to 
\[
\Psi^{out'}(\zero) = - \left( {\sqrt{d^{in}(v_1)}} (\II_{d^{in}(v_1)},\Psi^{in'}(v_1))_{\CO^{d^{in}(v_1)}}, ..., {\sqrt{d^{in}(v_N)}}(\II_{d^{in}(v_N)},\Psi^{in'}(v_N))_{\CO^{d^{in}(v_N)}}\right)^T,
\]
and 
\[
\psi^{in}(v_j) = \psi_j(0);
\]
which, in turns,  is equivalent to the Kirchhoff boundary conditions in $D(H^\ve)$.

The fact that the resolvent formula holds true for all $z \in \rho(H^\ve) \cap \rho(\mathring H^\ve)$, follows from \cite[Th. 2.19]{CFPrma18}.  

To prove the resolvent formula for $\WHR_z^{out}$, let $\chi\in\HH^{out}$ and set $\psi =  \WHR_z^{out} \chi$. By the first formula in \eqref{sigmaG}, one has 
\[
  \Psi(\zero) = - \WHP  (\WHP\MM^{out}\WHP)^{-1}\WHP\breve{G}^{out}_z \chi, 
\]
hence, $\WHP^\perp \Psi(\zero)=0$. Moreover, 
\[
\WHP \Psi'(\zero) =  \WHP \tau^{out} \psi  = \big(\ID - \WHP M_z^{out}\WHP(\WHP\MM^{out}\WHP)^{-1} \big) \WHP \breve{G}^{out}_z \chi = 0 . 
\]
Hence, the boundary conditions in $D(\WHH^{out})$ are satisfied, see Def. \ref{d:effect}. 

\subsection{Proof of Lemma \ref{l:Rve2}.\label{s:A2}}
Recall that we are denoting  by $\Kcal$ the space $\CO^{2N}$ or $\CO^N$ depending on if we are reasoning with operators in $\HH^{\ve}$, $\HH^{out}$ or $\HH^{in,\ve}$. 

\begin{Remark}\label{r:Minv}
By Identities  \eqref{MM} we infer 
\[
M_w^{-1} - M_{z}^{-1} = (z-w)    M_w^{-1} \breve G_{w} G_z  M_z^{-1} .
\]
Hence,  for  $\Im z\neq 0$, and for    any projection $P$ in $\Kcal$, and  $\uv\in P\Kcal$ 
\begin{equation}\label{astronomy}
\Im  (\uv,P\MM_z^{-1}P\uv)_{\Kcal} = \frac{1}{2i}\big(\uv,P(\MM_z^{-1} - \MM_{\bar z}^{-1})P\uv\big)_{\Kcal} 
= -\Im z \|G_z \MM_z^{-1}P\uv \|^2_{\HH}\neq  0 
\end{equation}
because $G_z \MM_z^{-1}$ is an injective map, being the composition of injective maps.  

Hence, the map $PM_z^{-1}P$ is invertible in $P \Kcal$. 
\end{Remark}

To prove that the map $ \MM^{in,\ve}_z\MM^{out}_z- \ID_N $ is invertible (the proof of the second statement in Eq. \eqref{MMMM} is analogous) note that it is enough to show that $ \MM^{in,\ve}_z- {\MM^{out}_z}^{-1}$ is invertible (because $M_z^{out}$ is). Let  $\uv\in\CO^N$, by Eq.s  \eqref{berry} and  \eqref{astronomy}
\[
\Im(\uv, \MM^{in,\ve}_z- {\MM^{out}_z}^{-1}\uv)_{\CO^N} = \Im z\big( \|G_z^{in,\ve}\uv\|_{\CO^N}^2 + \|G_z^{out} {M_z^{out}}^{-1}\uv\|^2_{\CO^N}\big) \neq 0.
\]

Formula \eqref{Rve_main}, comes from the block matrix inversion formula 
\[
\begin{pmatrix}
 \MM^{out}_z&-  \ID_N \\  \\ 
 -  \ID_N  & \MM^{in,\ve}_z
 \end{pmatrix}^{-1}
=
\begin{pmatrix}{\MM^{out}_z}^{-1} 	+{\MM^{out}_z}^{-1}
 \big( \MM^{in,\ve}_z-{\MM^{out}_z}^{-1}\big)^{-1} {\MM^{out}_z}^{-1} &{\MM^{out}_z}^{-1}
 \big( \MM^{in,\ve}_z-{\MM^{out}_z}^{-1}\big)^{-1}  \\  \\ 
\big( \MM^{in,\ve}_z-{\MM^{out}_z}^{-1}\big)^{-1}  {\MM^{out}_z}^{-1} &\big( \MM^{in,\ve}_z-{\MM^{out}_z}^{-1}\big)^{-1}   \end{pmatrix}
\]
together with the identities
\[
\begin{aligned}
{\MM^{out}_z}^{-1}
 \big( \MM^{in,\ve}_z-{\MM^{out}_z}^{-1}\big)^{-1} & = \big(\MM^{in,\ve}_z \MM^{out}_z- \ID_N\big)^{-1} \\ 
 \big( \MM^{in,\ve}_z-{\MM^{out}_z}^{-1}\big)^{-1}  {\MM^{out}_z}^{-1} &= \big(\MM^{out}_z\MM^{in,\ve}_z - \ID_N\big)^{-1},
 \end{aligned}
 \]
and 
 \[
 {\MM^{out}_z}^{-1} 	+{\MM^{out}_z}^{-1}
 \big( \MM^{in,\ve}_z-{\MM^{out}_z}^{-1}\big)^{-1} {\MM^{out}_z}^{-1}  = 
  \big( \MM^{in,\ve}_z\MM^{out}_z- \ID_N\big)^{-1} \MM^{in,\ve}_z. 
 \]

\section{Estimates on eigenvalues and eigenfunctions of $\mathring H^{in}$\label{s:appB}}
In this appendix we prove the following proposition on the asymptotic behavior of eigenvalues and eigenfunctions of $\mathring H^{in}$. 
%%%%%%%%%%%%%
%PROPOSITION
%%%%%%%%%%%%%
\begin{Proposition}
\label{p:Hdinoe} Recall that we denoted by $\{\la_n\}_{n\in\NA}$ the eigenvalues of the Hamiltonian $\mathring H^{in}$, and by $\{\varphi_n\}_{n\in\NA}$  a corresponding set of orthonormal eigenfunctions. There exists $n_0$ such that for any $n\geq n_0$:
\begin{equation}
\label{e:ae-o}
\lambda_n>n^2 C
\end{equation}
and 
\begin{equation}
\label{e:phinVj}
\sup_{x\in\GG^{in}}|\varphi_n(x)|\leq C
\end{equation}
for some positive  constant $C$ which does not depend on $n$.
\end{Proposition}
%%%%%%
%PROOF
%%%%%%
\begin{proof}
Claim \eqref{e:ae-o} is just the Weyl law. For $B^{in}=0$ a proof can be found in  \cite[Prop. 4.2]{bolte-endres-ahp09} (see also \cite{odzak-sceta-bmmss17}). For $B^{in}\neq 0$ bounded, claim \eqref{e:ae-o} can be deduced by a perturbative argument. 

To prove the bound  \eqref{e:phinVj} we follow the lines in the proof of
 Theorem A.1 in \cite{cur-wat-ip07}. For $b\in L^{\infty}(0,\ell)$ and
 real valued,  and $\la>0$ let $f
$ be the solution of  the equation
\begin{equation}
\label{e:eq}
 -f''+bf=\la f,
\end{equation}
with initial conditions $f(0)=f_0$ and $f'(0)=f_0'$. Then $f(x)$ can be
 written as  
\begin{equation}
\label{e:boys}
 f(x)=
\int_0^x\frac{\sin(\sqrt\la(x-y))}{\sqrt\la}b(y)f(y)dy
+f_0\cos(\sqrt\la x)+\frac{f_0'}{\sqrt\la}\sin(\sqrt\la x),
\end{equation}
from which it immediately follows that 
\[
|f(x)|\leq M + 
 \int_0^x\frac{1}{\sqrt\la}|b(y)||f(y)|dy,
\]
with 
\[
 M=|f_0|+\frac{|f_0'|}{\sqrt\la}.
\]
 Then from Gronwall's lemma, see, e.g. \cite[pg. 103]{Hormander97},  one has 
 \begin{equation}
 \label{e:boys2}
|f(x)|\leq M \exp\bigg(
 \int_0^x\frac{|b(y)|}{\sqrt\la}dy \bigg)\leq M \exp\bigg(
 \int_0^\ell |b(y)|dy \bigg),
\end{equation}
where we assumed  $\la>1$. By equation \eqref{e:boys} and by the estimate \eqref{e:boys2} it follows that 
\[
\bigg|f(x)-f_0\cos(\sqrt\la x)-\frac{f_0'}{\sqrt\la}\sin(\sqrt\la x)\bigg|
\leq
 M \exp\bigg(
 \int_0^\ell |b(y)|dy \bigg)
\int_0^x\frac{|b(y)|}{\sqrt\la}dy
\leq
C 
\bigg(\frac{|f_0|}{\sqrt\la}+\frac{|f_0'|}{\la}\bigg)
\]
where $C$ is a positive constant which does not depend on $\la$, $f_0$ and $f_0'$. We have then proved that 
\begin{equation}
\label{e:nir}
f(x)=f_0\cos(\sqrt\la x)+\frac{f_0'}{\sqrt\la}\sin(\sqrt\la x)+\OO_{L^\infty((0,\ell))}\bigg(\frac{|f_0|}{\sqrt\la}+\frac{|f_0'|}{\la}\bigg).
\end{equation}
Any component of the eigenfunction  $\varphi_n$ satisfies in the
 corresponding edge an equation of the form \eqref{e:eq}  with some
 initial data in $x=0$. Then the discussion on the function $f(x)$ above
 applies to all the components of the vector $\varphi_n $.  By
 the normalization condition $\|\varphi_n\|_{\HH^{in}}=1$ it
 follows that it must be $\|f\|_{L^2((0,l))}=C$, with $C\leq
 1$ (here $f$ denotes a generic component of $\varphi_n$, i.e., the restriction of $\varphi_n$ to a generic edge of $\GG^{in}$). Hence, from the identity 
\[\begin{aligned}
& \int_0^\ell \left|f_0\cos(\sqrt\la x)+\frac{f_0'}{\sqrt\la}\sin(\sqrt\la x) \right|^2 dx  \\ 
= &  
\frac{\ell}2\left(|f_0|^2+\frac{|f_0'|^2}{\la}\right) + 
\frac{\cos(2\sqrt\la \ell) - 1}{4\sqrt\la}\left(|f_0|^2-\frac{|f_0'|^2}{\la}\right) 
+ \frac{\Re(\bar f_0 f_0')}{\la} \sin^2(\sqrt\la \ell)
\end{aligned}
\]
one infers  
\[
C^2=\|f\|_{L^2((0,l))}^2=
\frac{\ell}2\left(|f_0|^2+\frac{|f_0'|^2}{\la}\right) +\OO\bigg(\frac{|f_0|^2}{\sqrt\la},
\frac{|f_0'|^2}{\la^{3/2}},
\frac{|f_0||f_0'|}{\la}\bigg).
\]
The latter estimate  implies that there exists $\tilde\la$  such that,  for all  $\la >\tilde \la$, the inequalities  $|f_0|\leq C_1$ and $|f_0'|/\sqrt\la\leq C_1$ hold true  for
 some positive constant $C_1$ which does  depend on $\la$. The 
 bounds  $|f_0|\leq C_1$ and $|f_0'|/\sqrt\la\leq C_1$, together with estimate \eqref{e:nir} and the fact that $\la_n\to+\infty$ for $n\to\infty$, imply \eqref{e:phinVj}.
\end{proof}

%\bibliographystyle{/Users/claudio/Dropbox/WIP-dropbox/bibliography/mybibstyle}
%\bibliography{app-bt} 

\begin{thebibliography}{999}

\bibitem{albeverio-cacciapuoti-finco:07}
Albeverio, S., Cacciapuoti, C., and Finco, D., \emph{Coupling in the singular
  limit of thin quantum waveguides}, J. Math. Phys. \textbf{48} (2007), 032103,
  21pp.

\bibitem{alb-pan-jpa05}
Albeverio, S. and Pankrashkin, K., \emph{A remark on {K}rein's resolvent
  formula and boundary conditions}, J. Phys. A: Math. Gen. \textbf{38} (2005),
  4859--4864.

\bibitem{mehmeti-ammari-nicaise:17}
{Ali Mehmeti}, F., Ammari, K., and Nicaise, S., \emph{Dispersive effects for
  the {S}chr\"odinger equation on the tadpole graph}, J. Math. Anal. Appl.
  \textbf{448} (2017), no.~1, 262--280.

\bibitem{berkolaiko-kuchment13}
Berkolaiko, G. and Kuchment, P., \emph{Introduction to quantum graphs},
  Mathematical Surveys and Monographs, vol. 186, American Mathematical Society,
  2013.

\bibitem{berkolaiko-latushkin-sukhtaiev:rxv18}
Berkolaiko, G., Latushkin, Y., and Sukhtaiev, S., \emph{Limits of quantum graph
  operators with shrinking edges}, arXiv:1806.00561 [math.SP] (2018), 30pp.

\bibitem{bolte-endres-ahp09}
Bolte, J. and Endres, S., \emph{The trace formula for quantum graphs with
  general self adjoint boundary conditions}, Ann. Henri Poincar\'e \textbf{10}
  (2009), 189--223.

\bibitem{bruining-geyler-pankrashkin:08}
Br\"uning, J., Geyler, V., and Pankrashkin, K., \emph{Spectra of self-adjoint
  extensions and applications to solvable {S}chr\"odinger operators}, Rev.
  Math. Phys. \textbf{20} (2008), no.~1, 1--70.

\bibitem{CFPrma18}
Cacciapuoti, C., Fermi, D., and Posilicano, A., \emph{{On inverses of
  Kre{\u\i}n's $\mathscr{Q}$-functions}}, Rend. Mat. Appl. \textbf{39} (2018),
  no.~7, 229--240.

\bibitem{cacciapuoti:17}
Cacciapuoti, C., \emph{{Graph-like asymptotics for the Dirichlet Laplacian in
  connected tubular domains}}, Analysis, Geometry and Number Theory \textbf{2}
  (2017), 25--58.

\bibitem{cacciapuoti-exner:07}
Cacciapuoti, C. and Exner, P., \emph{Nontrivial edge coupling from a
  {D}irichlet network squeezing: the case of a bent waveguide}, J. Phys. A
  \textbf{40} (2007), no.~26, F511--F523.

\bibitem{cacciapuoti-finco-aa10}
Cacciapuoti, C. and Finco, D., \emph{Graph-like models for thin waveguides with
  {R}obin boundary conditions}, Asymptot. Anal. \textbf{70} (2010), no.~3--4,
  199--230.

\bibitem{cheon-exner-turek:10}
Cheon, T., Exner, P., and Turek, O., \emph{Approximation of a general singular
  vertex coupling in quantum graphs}, Ann. Physics \textbf{325} (2010), no.~3,
  548--578.

\bibitem{cur-wat-ip07}
Currie, S. and Watson, B.~A., \emph{Inverse nodal problems for
  {S}turm-{L}iouville equations on graphs}, Inverse Problems \textbf{23}
  (2007), no.~5, 2029--2040. 

\bibitem{exner-manko13}
Exner, P. and Man'ko, S.~S., \emph{Approximations of quantum-graph vertex
  couplings by singularly scaled potentials}, J. Phys. A \textbf{46} (2013),
  no.~34, 345202, 17pp.

\bibitem{exner-manko:14}
Exner, P. and Man'ko, S.~S., \emph{Approximations of quantum-graph vertex
  couplings by singularly scaled rank-one operators}, Lett. Math. Phys.
  \textbf{104} (2014), no.~9, 1079--1094.

\bibitem{exner-post:05}
Exner, P. and Post, O., \emph{Convergence of spectra of graph-like thin
  manifolds}, J. Geom. Phys. \textbf{54} (2005), 77--115.

\bibitem{exner-post:aip08}
Exner, P. and Post, O., \emph{Quantum networks modelled by graphs}, AIP
  Conference Proceedings \textbf{998} (2008), no.~1, 1--17.

\bibitem{exner-post:09}
Exner, P. and Post, O., \emph{Approximation of quantum graph vertex couplings
  by scaled {S}chr{\"o}dinger operators on thin branched manifolds}, J. Phys. A
  \textbf{42} (2009), 415305, 22pp.

\bibitem{gohberg-goldberg-kaashoek-90}
Gohberg, I., Goldberg, S., and Kaashoek, M.~A., \emph{Classes of linear
  operators. {V}ol. {I}}, Operator Theory: Advances and Applications, vol.~49,
  Birkh\"{a}user Verlag, Basel, 1990.

\bibitem{golovaty-hryniv:09}
Golovaty, {\relax Yu D}. and Hryniv, R.~O., \emph{On norm resolvent convergence
  of {S}chr{\"o}dinger operators with $\delta'$-like potentials}, J. Phys. A:
  Math. Theor. \textbf{43} (2010), no.~15, 155204, 14pp.

\bibitem{gorbachuk-gorbachuk}
Gorbachuk, V.~I. and Gorbachuk, M.~L., \emph{Boundary value problems for
  operator differential equations}, Mathematics and its Applications (Soviet
  Series), vol.~48, Kluwer Academic Publishers Group, Dordrecht, 1991,
  Translated and revised from the 1984 Russian original.

\bibitem{Hormander97}
H{\"o}rmander, L., \emph{Lectures on nonlinear hyperbolic differential
  equations}, Math\'ematiques \& Applications (Berlin) [Mathematics \&
  Applications], vol.~26, Springer-Verlag, Berlin, 1997.

\bibitem{kostrykin-schrader-jpa99}
Kostrykin, V. and Schrader, R., \emph{Kirchhoff's rule for quantum wires}, J.
  Phys. A: Math. Gen. \textbf{32} (1999), 595--630.

\bibitem{kostrykin-schrader-cm06}
Kostrykin, V. and Schrader, R., \emph{Laplacians on metric graphs: eigenvalues,
  resolvents and semigroups}, Contemporary Mathematics \textbf{415} (2006),
  201--226.

\bibitem{manko:12}
Man'ko, S.~S., \emph{Schr\"odinger operators on star graphs with singularly
  scaled potentials supported near the vertices}, J. Math. Phys. \textbf{53}
  (2012), no.~12, 123521, 13pp.

\bibitem{manko:14}
Man'ko, S.~S., \emph{Quantum-graph vertex couplings: some old and new
  approximations}, Math. Bohem. \textbf{139} (2014), no.~2, 259--267.

\bibitem{manko:15}
Man'ko, S.~S., \emph{On {$\delta'$}-couplings at graph vertices}, Mathematical
  results in quantum mechanics, World Sci. Publ., Hackensack, NJ, 2015,
  pp.~305--313.

\bibitem{odzak-sceta-bmmss17}
Od{\v{z}}ak, A. and {\v{S}}\'ceta, L., \emph{On the {W}eyl law for quantum
  graphs}, Bull. Malays. Math. Sci. Soc. \textbf{42} (2019), no.~1, 119--131.

\bibitem{posilicano_jfa01}
Posilicano, A., \emph{A {K}re{\u\i}n-like formula for singular perturbations of
  self-adjoint operators and applications}, J. Funct. Anal. \textbf{183}
  (2001), 109--147.

\bibitem{pos-om08}
Posilicano, A., \emph{Self-adjoint extensions of restrictions}, Oper. Matrices
  \textbf{2} (2008), no.~4, 483--506.

\bibitem{post:06}
Post, O., \emph{Spectral convergence of quasi-one-dimensional spaces}, Ann.
  Henri Poincar{\'e} \textbf{7} (2006), 933--973.

\bibitem{post:book}
Post, O., \emph{Spectral analysis on graph-like spaces}, vol. 2039, Springer
  Science \& Business Media, 2012.

\bibitem{schmudgen}
Schm\"udgen, K., \emph{Unbounded self-adjoint operators on {H}ilbert space},
  Graduate Texts in Mathematics, vol. 265, Springer, Dordrecht, 2012.

\end{thebibliography}

\end{document}